\documentclass[12pt]{article}
\usepackage{amssymb}
\usepackage[top=1.20in,bottom=1.20in,left=1.00in,right=1.00in] {geometry}

\usepackage{fancyhdr}
\usepackage{amsmath}
\usepackage{bm}
\usepackage{esint}
\usepackage{amsfonts}
\usepackage{amsthm}
\usepackage{amssymb}
\usepackage{amsbsy}
\usepackage{verbatim}
\usepackage{graphicx}
\usepackage[caption=false,font=footnotesize]{subfig}
\usepackage{url}
\usepackage{mathrsfs}
\usepackage{color}
\usepackage{amstext}
\usepackage{stmaryrd}
\usepackage{enumerate}
\usepackage{algorithmic}
\usepackage{algorithm}
\usepackage{pifont}
\usepackage{prettyref}
\usepackage[percent]{overpic}
\usepackage{tcolorbox}
\usepackage{subfig}

\newtheorem{definition}{Definition}[section]

\newtheorem{lemma}{Lemma}[section]
\newtheorem{theorem}{Theorem}[section]

\allowdisplaybreaks
\newcommand\numberthis{\addtocounter{equation}{1}\tag{\theequation}}

\def\la{\left\langle}
\def\ra{\right\rangle}
\def\lb{\left(}
\def\rb{\right)}
\def\lcb{\left\{}
\def\rcb{\right\}}
\def\ln{\left\|}
\def\rn{\right\|}
\def\lsb{\left[}
\def\rsb{\right]}
\def\lab{\left|}
\def\rab{\right|}
\def\A{\mathcal{A}}
\def\C{\mathbb{C}}
\def\S{\mathcal{S}}
\def\P{\mathcal{P}}
\def\N{\mathcal{N}}

\def\H{\mathcal{H}}
\def\D{\mathcal{D}}

\def\O{\mathcal{Q}}

\def\R{\mathbb{R}}
\def\E{\mathbb{E}}
\def\I{\mathcal{I}}
\def\BZ{\bm{Z}}
\def\BH{\bm{H}}
\def\BL{\bm{L}}
\def\BU{\bm{U}}
\def\BV{\bm{V}}
\def\BB{\bm{B}}
\def\BF{\bm{F}}
\def\BC{\bm{C}}
\def\BD{\bm{D}}

\def\BW{\bm{W}}
\def\BQ{\bm{Q}}
\def\BLa{\bm{\Lambda}}
\def\BA{\bm{A}}
\def\BM{\bm{M}}
\def\BX{\bm{X}}
\def\BY{\bm{Y}}
\def\BQ{\bm{Q}}
\def\BR{\bm{R}}
\def\BS{\bm{\Sigma}}
\def\BG{\bm{G}}
\def\BXi{\bm{\xi}}
\def\BI{\bm{I}}
\def\Pr{\mathbb{P}}
\def\bx{\bm{x}}
\def\be{\bm{e}}

\def\bc{\bm{c}}
\def\bz{\bm{z}}
\def\by{\bm{y}}
\def\bb{\bm{b}}

\def\bw{\bm{w}}

\def\br{\bm{r}}
\def\bh{\bm{h}}
\def\bg{\bm{g}}

\DeclareMathOperator{\hard}{\mathcal{H}}
\DeclareMathOperator{\rank}{rank}
\DeclareMathOperator{\diag}{diag}
\DeclareMathOperator{\trace}{trace}
\DeclareMathOperator{\sigmin}{\sigma_{\min}}
\DeclareMathOperator{\sigmax}{\sigma_{\max}}
\DeclareMathOperator{\II}{{\textbf{\mbox{I}}}}

\begin{document}

\title{\textbf{Painless Breakups -- \\ Efficient Demixing of Low Rank Matrices}}


\author{Thomas Strohmer and Ke Wei\thanks{The authors acknowledge support from the NSF via grants DTRA-DMS 1322393 and DMS 1620455.}
 \\ Department of Mathematics \\ University of
California at Davis \\ Davis, CA 95616 \\ \{strohmer,kewei\}@math.ucdavis.edu}

\maketitle

\begin{abstract} 
Assume we are given a sum of linear measurements of $s$ different rank-$r$ matrices of the form $\by = \sum_{k=1}^{s} \A_k (\BX_k)$.
When and under which conditions is it possible to extract (demix) the individual matrices $\BX_k$ from the single measurement vector $\by$?
And can we do the demixing numerically efficiently?  
We  present two computationally efficient algorithms based on hard thresholding to solve this low rank demixing problem. We prove that under suitable
conditions these algorithms are guaranteed to converge to the correct solution at a linear rate. We discuss applications in connection with quantum tomography and the 
Internet-of-Things. Numerical simulations demonstrate the empirical  performance of the proposed
algorithms.
\end{abstract}

\section{Introduction}
\label{s:intro}


Demixing problems appear in a wide range of areas, including audio source separation~\cite{LXQZ09}, image processing~\cite{CE07}, communications engineering~\cite{WP98},
and astronomy~\cite{starck2005image}.  A vast amount of recent literature focuses on the demixing of signals  with very different properties. For instance one may be interested in the demixing of a signal that is the sum of spikes and sines~\cite{mccoy2014convexity} or the separation of a matrix that is the sum of a sparse matrix and a low rank matrix~\cite{clmw2006robust}.  
This paper focuses on the demixing problem where the signals to be separated all have {\em similar} properties, i.e., they are all low rank matrices. Assume we are given a sum of linear measurements
\begin{equation}
\by = \sum_{k=1}^{s} \A_k (\BX_k),
\label{lowranksum}
\end{equation}
where  $\{\A_k\}_{k=1}^s$ is a set of  linear operators from  $n\times n$ matrices to $m$-dimensional vectors and the constituents $\{\BX_k\}_{k=1}^s$ are unknown  rank-$r$ matrices. Our goal is to extract each constituent matrix $\BX_k$ from the single observation vector $\by$. We face different challenges in this homogeneous scenario, and one way to make the separation possible is by ensuring that the sensing operators $\A_k$ are sufficiently different from each other. The problem of demixing a sum of low rank matrices arises in  quantum tomography~\cite{deville2016concepts}, 
dictionary learning~\cite{huang2015convolutional} and wireless communications~\cite{WBSJ14}.

In quantum tomography one tries to reconstruct an unknown quantum state from experimental data~\cite{kliesch2017guaranteed}.
The state of a quantum system in quantum mechanics is often described by a density matrix, which is 
a positive semi-definite Hermitian matrice with unit trace. 
Many density matrices of interest are (or at least approximately) low rank---for instance pure states can be represented by rank-one matrices.
Hence, in quantum tomography the goal is to reconstruct low rank matrices from a set of linear measurements, and ideally by using as few measurements as possible. 
In~\cite{deville2016concepts},  a specific measurement protocol for quantum tomography is described in which  states get mixed together. 
In this case one has to reconstruct {\em and} demix quantum states (i.e., low rank matrices) from a set of  linear measurements.

Demixing problems of the form~\eqref{lowranksum} are also expected to arise in the future
Internet-of-Things (IoT). The IoT will connect billions of wireless devices, which is far more than the current
wireless systems can technically and economically accommodate. One of the many challenges in the design of the IoT
will be its ability to manage the massive number of \emph{sporadic traffic} generating devices which are  inactive most of 
the time, but regularly access the network for minor updates with no human interaction~\cite{WBSJ14}. 
It is common understanding among communication engineers that this traffic cannot be handled within the current 
random access procedures. Dimensioning the channel access according to classical information and communication
theory results in a severe waste of resources which does not scale towards the requirements of the IoT.
This means among others that the overhead caused by the exchange of certain types of information between
transmitter and receiver, such as channel estimation, assignment of data slots, etc, has to be avoided as much as possible.
Without explicit channel estimation the receiver needs to blindly deconvolve (to undo the effect of the channel) and demix the signals
arriving from many different devices. We will describe in Section~\ref{s:numerics} how this blind deconvolution-demixing problem can be phrased
as a demixing problem of rank-one matrices.


\subsection{State of the Art}

The low rank matrix demixing problem in~\eqref{lowranksum} is a generalization of the well-known low rank matrix recovery 
problem~\cite{rechtfazelparrilo2010nnm}. The task in low rank matrix recovery is to reconstruct a single low rank matrix $\BX$ from a few linear measurements $\by =  \A (\BX)$; that is, we
have $s=1$ in~\eqref{lowranksum}. This problem finds applications in a wide range of disciplines, including quantum tomography~\cite{kliesch2017guaranteed}, and image processing~\cite{rechtfazelparrilo2010nnm,zhou2015low}. Many different algorithms have been proposed for its solution~\cite{zhou2015low}, including convex optimization based methods~\cite{candesrecht2009mc} and non-convex methods such as thresholding 
algorithms~\cite{cai2010singular,tw2012nihtmc}. The thresholding-based algorithms proposed in this paper can be viewed as  natural generalizations of the latter algorithms.

There is a plethora of literature of various demixing  problems, such as ~demixing of a  low rank matrix and a sparse matrix, see for example~\cite{clmw2006robust}. 
However, most of these problems differ considerably from the one studied  in this paper, since they focus on the separation of two signals with complementary properties. Therefore we will not discuss them here in any further detail.  

Two of the first papers  that consider the demixing of a sum of low rank matrices are~\cite{WGMM13} and \cite{mccoy2013demixing}. In~\cite{mccoy2013demixing}, 
the authors consider a general demixing framework of the form
\begin{equation}
\by = \BA \left( \sum_{k=1}^s  \BU_k \BX_k + \bw_k \right),
\label{tropp}
\end{equation}
where $\BA$ is a right-invertible matrix, the  matrices $\{\BU_k\}_{k=1}^s$ are random unitary matrices, the matrices $\{\BX_k\}_{k=1}^s$ are the signals of interest and the vectors $\{\bw_k\}_{k=1}^s$ are noise.
In the  framework by McCoy and Tropp the signals $\BX_k$ are assumed to be highly structured, which includes the scenario where all 
the $\BX_k$ are rank-$r$ matrices.
Thus, in that case the setup~\eqref{tropp} is a special instance of~\eqref{lowranksum}. The focus of~\cite{mccoy2013demixing} is on theoretical
performance bounds of convex approaches for solving~\eqref{tropp}. In a nutshell the authors derive compelling quantitative phase transition bounds which show that demixing can succeed if and only if the dimension of the observation exceeds the total degrees of freedom present in the signals. 

A special case of \eqref{lowranksum}  where the constituents $\{\BX_k\}_{k=1}^s$ are rank-one matrices has been analyzed from a theoretical and a numerical
viewpoint in~\cite{LS15}. 
There, the authors investigate nuclear norm minimization (NNM) for the  demixing problem under structured measurements of practical interest,  and  show that $m\gtrsim s^2n$ number of measurements are sufficient for NNM to reliably extract each constituent $\BX_k$ from the single measurement  vector $\by$. It is also worth noting that an improvement of the theoretical analysis in~\cite{LS15} has been announced in~\cite{stoger2016blind}.

A limitation of the numerical methods proposed in~\cite{LS15} and \cite{mccoy2013demixing} is that the resulting semidefinite program is computationally rather expensive to solve 
for medium-size and large-size problems. Some applications require numerical efficiency, in which case a different numerical approach is needed. The goal of this paper is to develop  numerically efficient algorithms
that can solve the nonlinear low rank matrix demixing problem  without resorting to convex optimization, and meanwhile  to provide 
competitive theoretical recovery guarantees for the proposed algorithms. Closest to the goal of this paper is arguably a very recent paper by Ling and one of the authors~\cite{LS17}.
There, the authors consider the same setup as in~\cite{LS15}, and propose a non-convex regularized gradient-descent based method. The differences to this paper are that (i)~\cite{LS17} is specialized to the joint blind deconvolution-demixing setting and is not designed for~\eqref{lowranksum}, where the unknown matrices $\BX_k$ are general rank-$r$ matrices and the linear operators $\A_k$ are more general sensing matrices; (ii) while both algorithms fall in the realm of non-convex optimization,~\cite{LS17} uses a gradient descent method and this paper uses thresholding-based methods; (iii) the theoretical analysis in this paper does not apply to the joint blind deconvolution-demixing setting in~\cite{LS17}.

\subsection{Outline and Notation}

The remainder of the paper is organized as follows.  
The numerical algorithms  and their theoretical guarantees  for the demixing problem are presented in Section~\ref{s:setup}. We also introduce
an Amalgam form of  restricted isometry property in Section~\ref{s:setup} around which our theoretical analysis revolves.
In Section~\ref{s:numerics}  we test our algorithms on a variety of numerical examples.
The proofs of the theoretical results are presented in Section~\ref{s:proofs}. We conclude this paper with some potential future directions in Section~\ref{s:conclusion}.

Throughout the paper we use the following notational conventions. We denote vectors by bold lowercase letters and matrices by bold uppercase letters.  In particular, we fix $\{\BX_{k}\}_{k=1}^s$ as the target matrices and $\{\BX_{k,l}\}_{k=1}^s$ as the iterates of the algorithms. For a matrix $\BZ$, we use $\ln\BZ\rn$ and $\ln\BZ\rn_F$ to denote its spectral norm and Frobenius norm, respectively. 
For both vectors and matrices, $\bz^T$ and $\BZ^T$ denote their transpose while $\bz^*$ and $\BZ^*$ denote their conjugate transpose. The inner product of two matrices $\BZ_1$ and $\BZ_2$ is defined as $\la \BZ_1,\BZ_2\ra=\trace(\BZ_1^*\BZ_2)$.

 \section{Amalgam-RIP and Algorithms}
\label{s:setup}
\subsection{Problem Setup and Amalgam-RIP}
As outlined in the introduction, we want to solve the following demixing problem:
\begin{align*}
\text{Find all rank-$r$ matrices $\BX_k$, \qquad given $\by = \sum_{k=1}^s\A_k(\BX_k),$} \numberthis\label{eq:linear_measurements}
\end{align*}
where  $\{\A_k\}_{k=1}^s$ is a set of linear operators mapping  $n\times n$ matrices to $m$-dimensional vectors. Let  $\{\BA_{k,p}~|~1\leq k\leq s, 1\leq p\leq m\}$ be a set of measurement matrices. We can write $\A_k(\BZ)$ explicitly as  
\begin{align*}
\A_k(\BZ) =\frac{1}{\sqrt{m}} \begin{bmatrix}
\la \BA_{k,1},\BZ\ra\\
\vdots\\
\la \BA_{k,m},\BZ\ra
\end{bmatrix}.\numberthis\label{eq:linear_operator}
\end{align*}
For conciseness, we only present  our algorithms and theoretical results for $n\times n$ real matrices,  but it is straightforward to modify them for $n_1\times n_2$ complex matrices.



We will propose two iterative hard thresholding algorithms for low rank matrix demixing. A question of central importance  is how many measurements are needed so that the algorithms can successfully extract all the constituents $\{\BX_k\}_{k=1}^s$ from $\by$. Since each $\BX_k$  is determined by $r(2n-r)$ parameters \cite{candesplan2009oracle}, we need at least $m \ge sr(2n-r)$ many measurements for the problem to be well-posed.  One of our goals
is to show that the proposed algorithms can succeed already if the number of measurements is close to this information-theoretic minimum. As is common for hard thresholding algorithms in compressed sensing and low rank matrix recovery, the  convergence analysis here will rely on some form of restricted isometry property (RIP). The notion most fitting to the demixing setting requires the RIP in an amalgamated form.

\begin{definition}[Amalgam-RIP] \label{def:rip}
The set of linear operators $\{\A_k\}_{k=1}^s$ satisfy the Amalgam-RIP (ARIP) with the parameter $\delta_r$ if 
\begin{align*}
(1-\delta_r)\sum_{k=1}^s\ln\BZ_k\rn_F^2\leq \ln\sum_{k=1}^s\A_k(\BZ_k)\rn^2\leq (1+\delta_r)\sum_{k=1}^s\ln\BZ_k\rn_F^2\numberthis\label{eq:ARIP_def}
\end{align*}
holds for all the matrices $\BZ_k,~k=1,\cdots,s$ of rank at most $r$.
\end{definition}

Fix an index $k$ and take $\BZ_{k'}=0$ for all $k'\ne k$. The ARIP implies that 
\begin{align*}
(1-\delta_r)\ln\BZ_k\rn_F^2\leq \ln\A_k(\BZ_k)\rn^2\leq (1+\delta_r)\ln\BZ_k\rn_F^2
\end{align*}
hold for all matrices $\BZ_k$ of rank at most $r$, which is indeed the RIP introduced in \cite{rechtfazelparrilo2010nnm} for low rank matrix recovery.  
When the measurements matrices in \eqref{eq:linear_operator} have i.i.d $\N(0,1)$ entries, we can show that $\{\A_k\}_{k=1}^s$ satisfy the ARIP with overwhelmingly high probability  provided the number of measurements is proportional to the  number of degrees of freedom within the constituents $\{\BX_k\}_{k=1}^s$.
\begin{theorem}\label{thm:ARIP}
If $\{\BA_{k,p}~|~1\leq k\leq s, 1\leq p\leq m\}$ is a set of standard Gaussian matrices, then  $\{\A_k\}_{k=1}^s$ satisfy the ARIP with the parameter $\delta$  with probability at least $1-2\exp(-m\delta^2/64)$  provided 
\begin{align*}
m\geq C\delta^{-2}(2n+1)rs\log(1/\delta),
\end{align*}
where $C>0$ is an absolute numerical constant.
\end{theorem}
\subsection{Iterative Hard Thresholding Algorithms}
In this subsection, we present two iterative hard thresholding algorithms for the low rank demixing problem. The algorithms are developed by targeting the following rank constraint problem directly
\begin{align*}
\min_{\{\BZ_k\}_{k=1}^s}\ln\by- \sum_{k=1}^s\A_k(\BZ_k)\rn^2\quad\mbox{subject to}\quad\rank(\BZ_k)\leq r.
\end{align*}

The first algorithm, simply referred to as iterative hard thresholding (IHT), is presented in Algorithm~\ref{alg:iht}. In each iteration, the algorithm first computes the residual  and then updates each constituent separately by projected gradient descent. The search direction with respect to each constituent is computed as the negative gradient descent direction of the objective function when the other constituents are fixed. The hard thresholding operator $\H_r$ in  the algorithm returns the best rank $r$ approximation of a matrix. The stepsize within the line search can be either fixed or computed adaptively. In Algorithm~\ref{alg:iht}, we compute  it as the steepest descent stepsize along the projected gradient descent direction $\P_{T_{k,l}}( \BG_{k,l})$, where $T_{k,l}$ is the tangent space of the rank $r$ matrix manifold at the  current estimate $\BX_{k,l}$. Let  $\BX_{k,l}=\BU_{k,l}\BS_{k,l}\BV_{k,l}^*$ be the singular value  decomposition of $\BX_{k,l}$. The tangent space $T_{k,l}$ consists of matrices which share the same column or row subspaces with $\BX_{k,l}$ \cite{AbMaSe2008manifold},
\begin{align*}
T_{k,l} = \{\BU_{k,l}\BZ_1^*+\BZ_2\BV_{k,l}^*~|~\BZ_1\in\R^{n\times r},~\BZ_2\in\R^{n\times r}\}.\numberthis\label{eq:tangent}
\end{align*}
Note that the inner loop within the algorithm is fully parallel since we update each constituent separately.
\begin{algorithm}[!ht]
\caption{Iterative Hard Thresholding for Low Rank Demixing}\label{alg:iht}
\begin{algorithmic}
\STATE\textbf{Initialization}: $\BX_{k,0}=\BU_{k,0}\BS_{k,0}\BV_{k,0}^* = \hard_r(\A_k^*(\by))$
\FOR{$l=0,1,\cdots$}
\STATE  $\br_{l}=\by-\sum_{k=1}^s\A_k(\BX_{k,l})$
\FOR[\emph{Fully Parallel}]{$k=1,\cdots,s$} 
\STATE 1. $\BG_{k,l}=\A_k^*(\br_l)$
\STATE 2. $\alpha_{k,l}= \frac{\|\P_{T_{k,l}}( \BG_{k,l})\|_F^2}{\| \A_k\P_{T_{k,l}}( \BG_{k,l})\|_2^2}$
\STATE 3. $\BX_{k,l+1}=\hard_r(\BX_{k,l}+ \alpha_{k,l}\BG_{k,l})$
\ENDFOR
\ENDFOR
\end{algorithmic}
\end{algorithm} 

\begin{theorem}\label{thm:iht}
Assume the linear operators $\{\A_k\}_{k=1}^s$ satisfy the ARIP with the parameter $\delta_{3r}$. 
Define
\begin{align*}
\gamma_{1} = \frac{4\delta_{3r}}{1-\delta_{3r}}.
\end{align*}
If $\gamma_{1}<1$, then the iterates of IHT satisfy 
\begin{align*}\ln
\begin{bmatrix}
\BX_{1,l+1}-\BX_1\\
\vdots\\
\BX_{s,l+1}-\BX_s
\end{bmatrix}\rn_F
\leq 
\gamma_{1} \ln
\begin{bmatrix}
\BX_{1,l}-\BX_1\\
\vdots\\
\BX_{s,l}-\BX_s
\end{bmatrix}\rn_F.
\end{align*}
It follows immediately  that the iterates of IHT converge linearly to the underlying constituents provided  $\delta_{3r}<1/5$. 
\end{theorem}
If $\{\A_k\}_{k=1}^s$ consist of standard Gaussian measurement matrices, then Theorem~\ref{thm:iht} together with Theorem~\ref{thm:ARIP}  implies that the number of necessary measurements for  IHT to achieve successful recovery is $O(s nr)$ which is optimal.  
In addition, we can establish the robustness of IHT against additive noise by considering the model $\by = \sum_{k=1}^s\A_k(\BX_k)+\be$, where $\be$ represents a noise term. 
\begin{theorem}\label{thm:iht_noise}
Assume $\by = \sum_{k=1}^s\A_k(\BX_k)+\be$ and the linear operators $\{\A_k\}_{k=1}^s$ satisfy the ARIP with parameter $\delta_{3r}$.   If $\gamma_{1}<1$, then the iterates of IHT satisfy
\begin{align*}
\ln
\begin{bmatrix}
\BX_{1,l}-\BX_1\\
\vdots\\
\BX_{s,l}-\BX_s
\end{bmatrix}\rn_F\leq\gamma_1^l\ln
\begin{bmatrix}
\BX_{1,0}-\BX_1\\
\vdots\\
\BX_{s,0}-\BX_s
\end{bmatrix}\rn_F+\frac{\xi}{1-\gamma_1}\ln\be\rn,
\end{align*}
where
\begin{align*}
\xi=\frac{2\sqrt{1+\delta_{3r}}}{1-\delta_{2r}}.
\end{align*}
\end{theorem}

The application of the hard thresholding operator in Algorithm~\ref{alg:iht} requires the singular value decomposition (SVD) of  an $n\times n$ matrix in each iteration, which is computationally expensive for unstructured matrices. Inspired by the work of Riemannian optimization for low rank matrix reconstruction in \cite{bart2012riemannian,CGIHT_dense, CGIHT_entry}, we propose to accelerate IHT by updating each constituent along a projected gradient descent direction, see fast iterative hard thresholding (FIHT) described in Algorithm~\ref{alg:fiht}. The key difference between Algorithms~\ref{alg:iht} and \ref{alg:fiht} lies in Step~3.  Instead of updating $\BX_{k,l}$ along the gradient descent direction $\BG_{k,l}$ as in IHT, FIHT updates $\BX_{k,l}$ along the projected gradient descent direction $\P_{T_{k,l}}( \BG_{k,l})$, where $T_{k,l}$ is the tangent space defined in \eqref{eq:tangent}, and then followed by the hard thresholding operation.

\begin{algorithm}[!ht]
\caption{Fast Iterative Hard Thresholding  for Low Rank Demixing}\label{alg:fiht}
\begin{algorithmic}
\STATE\textbf{Initialization}: $\BX_{k,0}=\BU_{k,0}\BS_{k,0}\BV_{k,0}^* = \hard_r(\A_k^*(\by))$
\FOR{$l=0,1,\cdots$}
\STATE  $\br_{l}=\by-\sum_{k=1}^s\A_k(\BX_{k,l})$
\FOR[\emph{Fully Parallel}]{$k=1,\cdots,s$} 
\STATE 1. $\BG_{k,l}=\A_k^*(\br_l)$
\STATE 2. $\alpha_{k,l}= \frac{\|\P_{T_{k,l}}( \BG_{k,l})\|_F^2}{\| \A_k\P_{T_{k,l}}( \BG_{k,l})\|_2^2}$
\STATE 3. $\BX_{k,l+1}=\hard_r(\BX_{k,l}+\alpha_{k,l}\P_{T_{k,l}}( \BG_{k,l}))$
\ENDFOR
\ENDFOR
\end{algorithmic}
\end{algorithm} 

Let $\BW_{k,l}=\BX_{k,l}+\alpha_{k,l}\P_{T_{k,l}}( \BG_{k,l})$. One can easily observe that $\BW_{k,l}\in T_{k,l}$ and $\rank(\BW_{k,l})\leq 2r$. Moreover, $\BW_{k,l}$ obeys the following decomposition \cite{CGIHT_dense}:
\begin{align*}
\BW_{k,l} = \begin{bmatrix}\BU_{k,l} & \BQ_1\end{bmatrix}\BM_{k,l}\begin{bmatrix}
\BV_{k,l}^*\\\BQ_2^*
\end{bmatrix}
\end{align*}
where $\BQ_1$ and $\BQ_2$ are two $n\times r$ orthonormal matrices such that $\BQ_1\perp\BU_{k,l}$ and $\BQ_2\perp\BV_{k,l}$,  and $\BM_{k,l}$ is a $2r\times 2r$ matrix of the form
\begin{align*}
\BM_{k,l} = \begin{bmatrix}
\BS_{k,l}+\BXi_{k,l} & \BR_1^*\\
\BR_2^* & 0
\end{bmatrix}.
\end{align*} 
Consequently, the SVD of $\BW_{k,l}$ can be obtained from the SVD of the $2r\times 2r$ matrix $\BM_{k,l}$. Therefore, the introduction of the extra projection $\P_{T_{k,l}}$ for the search direction can reduce the computational complexity of the partial SVD of $\BW_{k,l}$ from $O(n^2r)$ flops to $O(nr^2+r^3)$ flops. The recovery guarantee of FIHT can also be established in terms of the ARIP.

\begin{theorem}\label{thm:fiht}
Assume the linear operators $\{\A_k\}_{k=1}^s$ satisfy the ARIP with the parameter $\delta_{3r}$. 
Define 
\begin{align*}
\gamma_2 = 2\lb\frac{2\delta_{2r}}{1-\delta_{2r}}+\frac{\delta_{3r}}{1-\delta_{2r}}+2\delta_{3r}\sqrt{rs}\frac{\sigmax}{\sigmin}\rb,
\end{align*}
where $\sigmax:=\max_k\sigmax(\BX_k)$ and $\sigmin:=\max_k\sigmin(\BX_k)$.
If $\gamma_2<1$, then 
the iterates of Algorithm~\ref{alg:fiht} satisfy
\begin{align*}
\ln \begin{bmatrix}
\BX_{1,l+1}-\BX_1\\
\vdots\\
\BX_{s,l+1}-\BX_s
\end{bmatrix}\rn_F\leq\gamma_2\ln \begin{bmatrix}
\BX_{1,l}-\BX_1\\
\vdots\\
\BX_{s,l}-\BX_s
\end{bmatrix}\rn_F.
\end{align*}
It follows that  the iterates of FIHT converge linearly to the underlying constituents provided  
\begin{align*}
\delta_{3r}\lesssim\frac{\sigmin}{\sigmax}\frac{1}{\sqrt{rs}}.
\end{align*}
\end{theorem}
Assume $\{\A_k\}_{k=1}^s$ consist of standard Gaussian matrices.  In contrast to the $O(srn)$ necessary measurements for IHT to achieve successful recovery, Theorem~\ref{thm:fiht} implies that FIHT requires $O((\sigmax/\sigmin)^2s^2nr^2)$ measurements which is not optimal.  However, numerical simulations in Section~\ref{s:numerics} suggest that FIHT needs  fewer measurements than IHT to be able to successfully reconstruct $s$ rank-$r$ matrices. 

\section{Numerical Simulations and Applications}
\label{s:numerics}
We evaluate the empirical performance of IHT and FIHT on simulated problems as well as application examples from quantum states demxing and the Internet of Things. In the implementation IHT and FIHT are terminated when a maximum of $500$ iterations is met or the relative residual is small,
\begin{align*}
\frac{\ln\by-\sum_{k=1}^s\A_k(\BX_{k,l})\rn}{\ln\by\rn}\leq 10^{-4}.
\end{align*}
They are considered to have successfully recovered a set of matrices $\{\BX_k\}_{k=1}^s$  if the returned approximations $\{\BX_{k,l}\}_{k=1}^s$ satisfy
\begin{align*}
\frac{\sqrt{\sum_{k=1}^s\ln\BX_{k,l}-\BX_k\rn_F^2}}{\sqrt{\sum_{k=1}^s\ln\BX_k\rn_F^2}}\leq 10^{-2}.
\end{align*}
All the random tests are repeated {\em ten} times in this section.
\subsection{Phase Transitions under Gaussian Measurements}\label{sec:num_gaussian}
The ARIP-based theoretical recovery guarantees in Theorems~\ref{thm:iht} and \ref{thm:fiht} are worst-case analysis which are uniform over all rank $r$ matrices. However, these conditions are highly pessimistic when compared with average-case empirical observations. In this subsection we investigate the empirical  recovery performance of IHT and FIHT under Gaussian measurements. 
The tests are conducted on $n\times n$ rank-$r$ matrices with $n=50$ and $r=5$, and  the number of constituents $s$ varies  from $1$ to $7$. For a fixed value of constituents,  
a number of equispaced values of $m$ are tested.

We first test IHT and FIHT on well-conditioned matrices which are formed via $\BX_k=\BL_k\BR_k^*$, where $\BL_k\in\R^{n\times r}$, $\BR_k\in\R^{n\times r}$, and their entries are drawn  i.i.d. from the standard normal distribution. The phase transitions of IHT and FIHT on well-conditioned matrices are presented in Figures~\ref{fig:IHT_Gaussian} and \ref{fig:FIHT_Gaussian}, respectively. Figure~\ref{fig:IHT_Gaussian} shows a linear correlation between the number of measurements $m$ and the number of constituents  $s$ for the successful recovery of IHT. Though the recovery guarantee of FIHT in Theorem~\ref{thm:fiht} is worse than that of IHT in Theorem~\ref{thm:iht}, Figure~\ref{fig:FIHT_Gaussian} shows that FIHT requires fewer measurements than IHT to achieve successful recovery with high probability. In particular, $m\approx 3.3 sr(2n-r)$ is sufficient for FIHT to successfully reconstruct a set of $s$ low rank matrices for all the tested values of $s$.

\begin{figure}[!ht]
\centering
\subfloat[]{\includegraphics[width=0.35\textwidth]{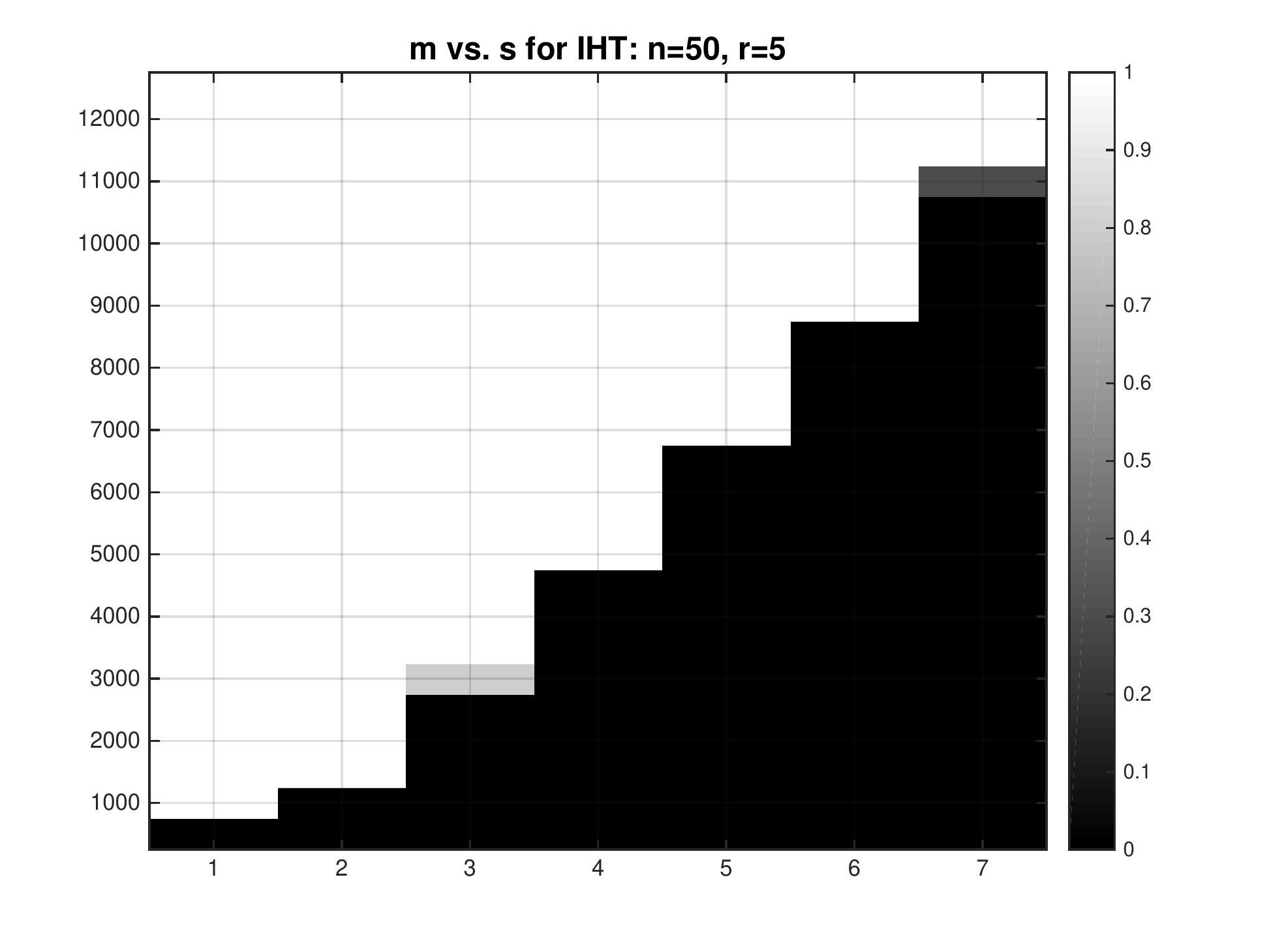}\label{fig:IHT_Gaussian}}
\subfloat[]{\includegraphics[width=0.35\textwidth]{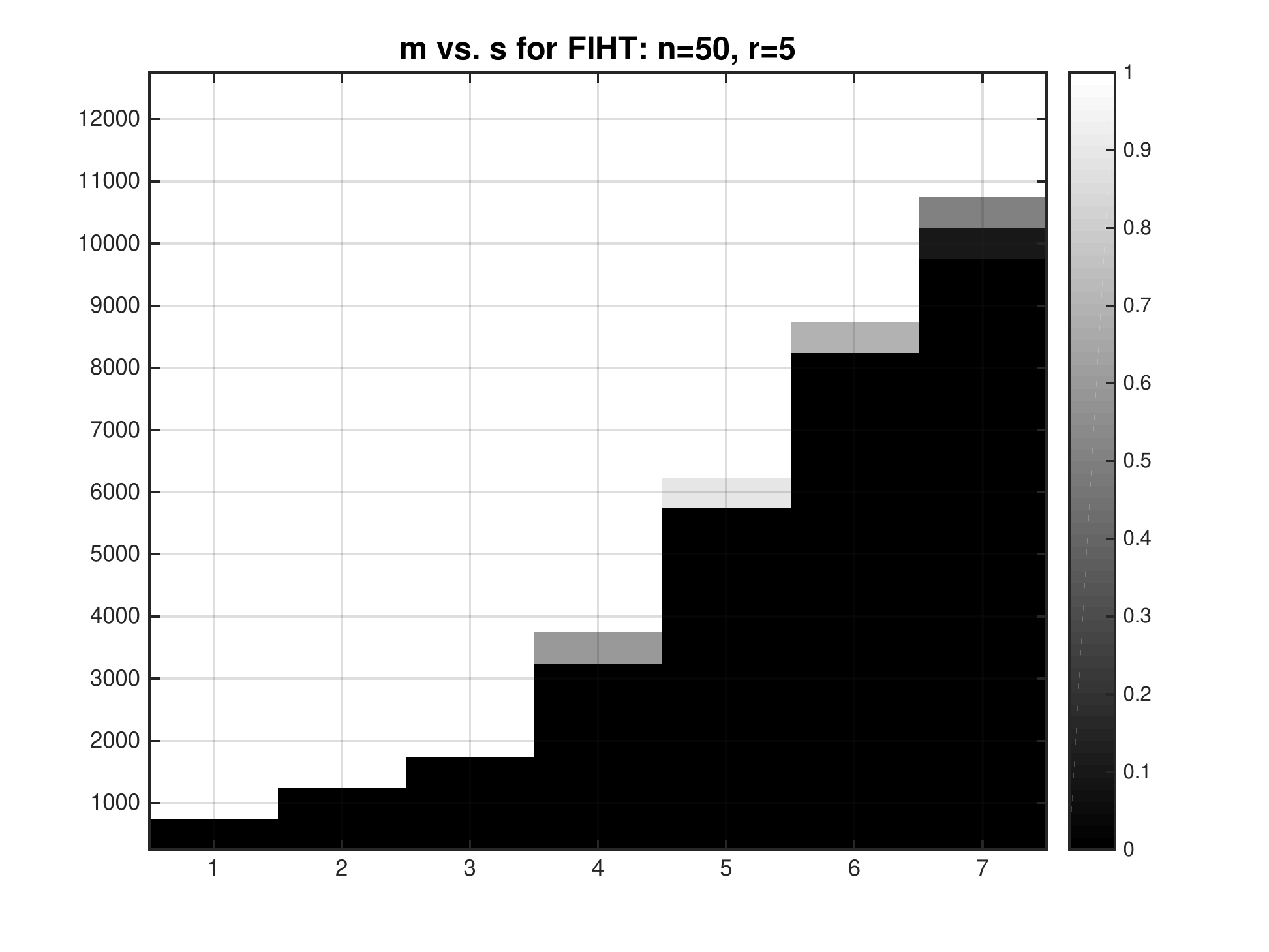}\label{fig:FIHT_Gaussian}}
\subfloat[]{\includegraphics[width=0.35\textwidth]{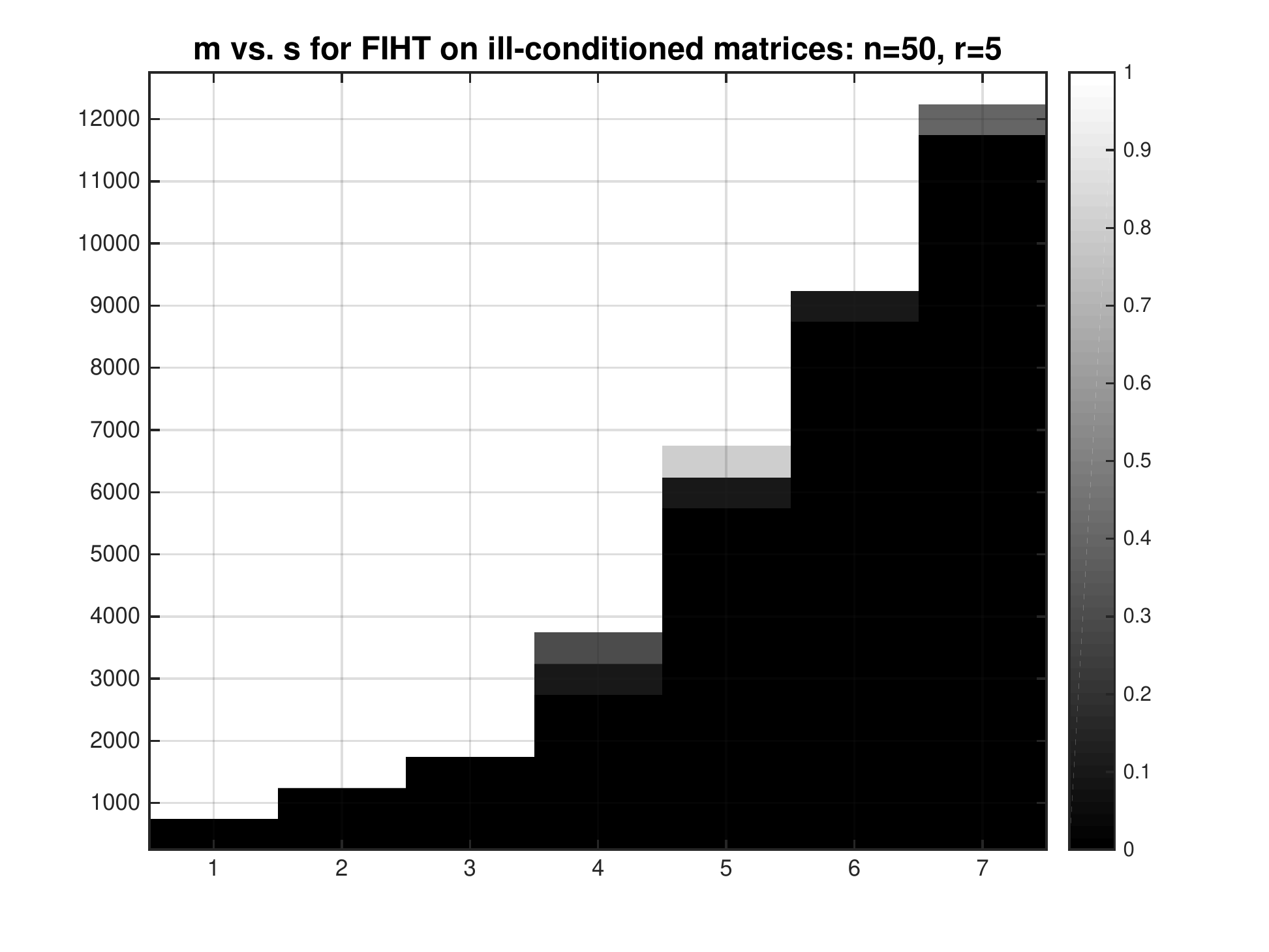}\label{fig:FIHT_Gaussian_Ill}}
\caption{Empirical phase transitions of IHT and FIHT under Gaussian measurements.}
\label{fig:phase_gaussian}
\end{figure}

To explore how the condition number of the test matrices impact the recovery of FIHT, we further test FIHT on ill-conditioned matrices. The matrices are computed as $\BX_k=\BU_k\BS\BV_k^*$, where $\BU_k$ and $\BV_k$ are two $50$ by $5$  random orthonormal matrices, and $\BS$ is a $5$ by $5$ diagonal matrix with the diagonal entries being $5$ uniformly spaced values from $1$ to $1000$. So the condition number of the test matrices is $1000$. The phase transition plot of FIHT on ill-conditioned matrices is presented in Figure~\ref{fig:FIHT_Gaussian_Ill}. The figure shows that FIHT requires more measurements to reconstruct the test ill-conditioned matrices than to reconstruct the  well-conditioned matrices when $s\geq 5$. However, FIHT is still able to reconstruct all the test ill-conditioned matrices when $m\approx 3.8 sr(2n-r)$. Hence, there is no strong evidence from Figure~\ref{fig:FIHT_Gaussian_Ill} to support that the number of measurements for the successful recovery of FIHT relies on the square of the condition number of the target matrices, which may suggest the possibility of improving the result in Theorem~\ref{thm:fiht}.

In the reminder of this section, we focus on the empirical performance of FIHT for its superiority and flexibility. 
Though we have only established the recovery guarantee of FIHT for real matrices under the Gaussian measurements, the algorithm is equally effective for complex matrices and other types of measurements. In the next two subsections, we study the applications of FIHT  in quantum states demixing and the Internet of Things.

\subsection{An Application in Quantum States Demixing}\label{sec:num_pauli}
As stated in the introduction, the state of a quantum system can be represented by a low rank, positive semidefinite  density matrix of unit trace.  Next, we present a stylized application of FIHT to the demixing of quantum states. 
When the target constituents are $n$ by $n$ rank-$r$ positive semidefinite matrices of unit  trace, FIHT in Algorithm~\ref{alg:fiht} can be modified accordingly to preserve the matrix structures, see Algorithm~\ref{alg:fiht_psd} for  complete description. 

\begin{algorithm}[!ht]
\caption{Fast Iterative Hard Thresholding  for Quantum States Demixing}\label{alg:fiht_psd}
\begin{algorithmic}
\STATE\textbf{Initialization}: $\BX_{k,0}=\BU_{k,0}\BLa_{k,0}\BU_{k,0}^* = \P_{\Delta}\lb\P_{\Pi}(\A_k^*(\by))\rb$
\FOR{$l=0,1,\cdots$}
\STATE  $\br_{l}=\by-\sum_{k=1}^s\A_k(\BX_{k,l})$
\FOR[\emph{Fully Parallel}]{$k=1,\cdots,s$} 
\STATE 1. $\BG_{k,l}=\A_k^*(\br_l)$
\STATE 2. $\alpha_{k,l}= \frac{\|\P_{T_{k,l}}( \BG_{k,l})\|_F^2}{\| \A_k\P_{T_{k,l}}( \BG_{k,l})\|_2^2}$
\STATE 3. $\BX_{k,l+1}=\P_{\Delta}\lb\P_{\Pi}(\BX_{k,l}+\alpha_{k,l}\P_{T_{k,l}}( \BG_{k,l}))\rb$
\ENDFOR
\ENDFOR
\end{algorithmic}
\end{algorithm} 

Let $\BX_{k,l}$ be the current rank-$r$ estimate which is Hermitian, positive semidefinite, and so admits an eigenvalue decomposition (equivalent to its SVD) $\BX_{k,l}=\BU_{k,l}\BLa_{k,l}\BU_{k,l}^*$ with non-negative eigenvalues. There are two key differences between Algorithms~\ref{alg:fiht} and \ref{alg:fiht_psd}.  Firstly, the tangent space can be defined as 
\begin{align*}
T_{k,l} = \{\BU_{k,l}\BZ^*+\BZ\BU_{k,l}^*~|~\BZ\in\C^{n\times r}\}.
\end{align*}
so that all the matrices in $T_{k,l}$ are Hermitian. Therefore after $\BX_{k,l}$ is updated along the projected gradient descent direction, $\BW_{k,l}=\BX_{k,l}+\alpha_{k,l}\P_{T_{k,l}}( \BG_{k,l})$ remains Hermitian. In addition, the eigenvalue decomposition of $\BW_{k,l}$ can also be computed using $O(nr^2)$ flops as in the computation of the SVD of the non-symmetric matrices in Algorithm~\ref{alg:fiht}, since $\BW_{k,l}$ has the following decomposition:
\begin{align*}
\BW_{k,l} = \begin{bmatrix}\BU_{k,l} & \BQ\end{bmatrix}\BM_{k,l}\begin{bmatrix}
\BU_{k,l}^*\\\BQ^*\end{bmatrix},
\end{align*}
where $ \begin{bmatrix}\BU_{k,l} & \BQ\end{bmatrix}$ is an $n\times 2r$ orthonormal matrix and $\BM_{k,l}$ is a  $2r\times 2r$ Hermitian matrix of the form
\begin{align*}
\BM_{k,l} = \begin{bmatrix}
\BLa_{k,l}+\BXi_{k,l} & \BR^*\\
\BR^* & 0
\end{bmatrix}.
\end{align*} Secondly, the hard thresholding operator $\H_r$ is replaced by the projection $\P_{\Pi}$, followed by another projection $\P_{\Delta}$, where $\P_{\Pi}$ projects $\BW_{k,l}$ onto the set of low rank positive semidefinite matrices
\begin{align*}
\Pi = \{\BZ\in\C^{n\times n}~|~\BZ=\BZ^*,\rank(\BZ)\leq r,~\mbox{ and }\BZ\succeq 0\};
\end{align*} 
and $\P_{S}$ projects   the eigenvalues of $\P_{\Pi}(\BW_{k,l})$ onto the simplex
\begin{align*}
\Delta = \lcb\bx\in\R^r~|~\sum_{i=1}^rx_i = 1,~x_i\geq 0\rcb.
\end{align*}

Notice that the inertia of a Hermitian matrix $\bm{H}$, denoted by $\mbox{inertia}(\bm{H})$, is a  triple of integer numbers ($i_p,i_n,i_z$), where $i_p$, $i_n$, $i_z$ are the number of positive, negative, and zero eigenvalues of $\bm{H}$. From \cite[Lemma~A.15]{FGW2002NonOp}, we know that $\mbox{inertia}(\BM_{k,l})=\mbox{inertia}(\bm{N}^*\BM_{k,l}\bm{N})+(t,t,r-t)$, where $t=\rank(\BR)$ and $\bm{N}$ is an $r\times (r-t)$ matrix whose columns forms a null space of $\BR$. This implies that $\BM_{k,l}$ has at most $r$ positive eigenvalues, so does $\BW_{k,l}$. Therefore by \cite[Theorem~2.1]{Higham1988}, we can compute $\P_{\Pi}(\BW_{k,l})$  by only keeping the positive eigenvalues and the corresponding eigenvectors of $\BM_{k,l}$. To compute $\P_S(\cdot)$, we use an algorithm proposed in \cite{chen2011projection}, in which the computational cost is dominated by sorting a vector of length $r$. 


\begin{figure}[!ht]
\centering
\subfloat[]{\includegraphics[width=0.45\textwidth]{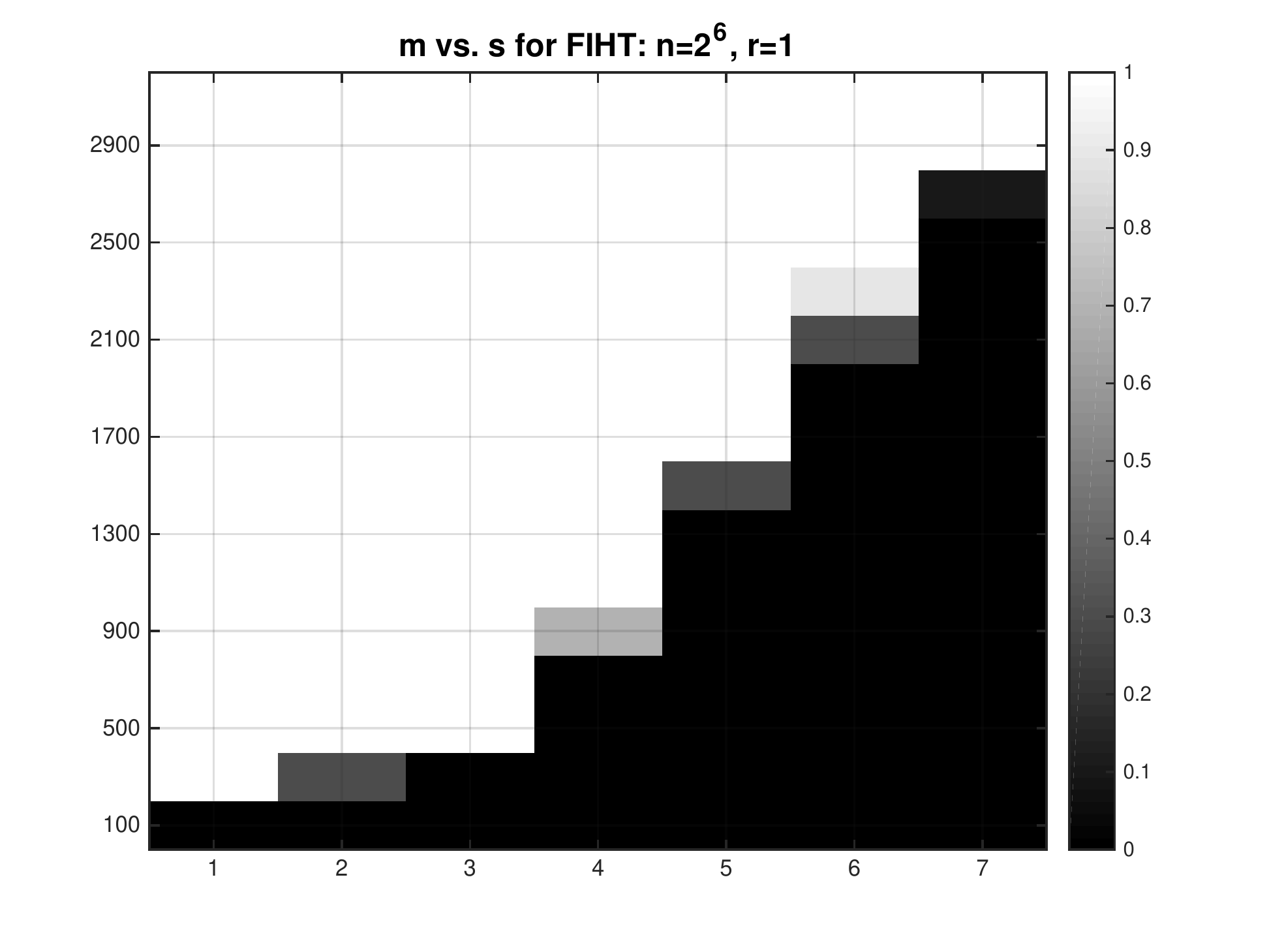}\label{fig:FIHT_Pauli_RankOne}}
\subfloat[]{\includegraphics[width=0.45\textwidth]{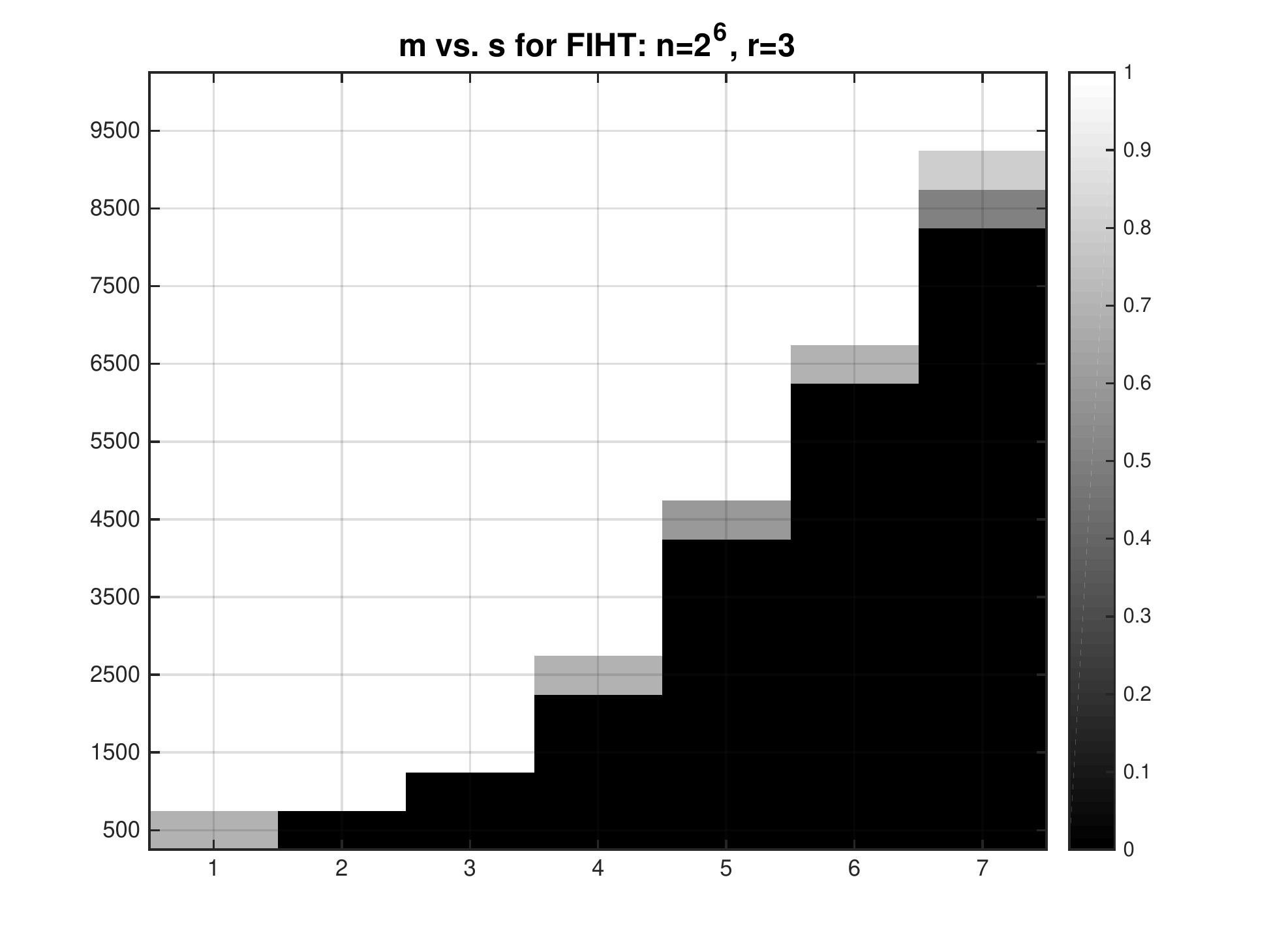}\label{fig:FIHT_Pauli_Rank3}}
\caption{Empirical phase transitions of  FIHT for quantum states demixing.}
\label{fig:phase_pauli}
\end{figure}

We evaluate the recovery performance of FIHT for quantum states demixing under the Pauli measurements, where the measurement matrices in \eqref{eq:linear_operator} are tensor products of the Pauli matrices. That is, 
\begin{align*}
\BA_{k,p}=2^{-q/2}\otimes_{i=1}^q\sigma_a,
\end{align*}
where 
\begin{align*}
\sigma_1=\begin{bmatrix}
0 & 1 \\1 &0
\end{bmatrix},~
\sigma_2=\begin{bmatrix}
0 & -i \\i &0
\end{bmatrix},~
\sigma_3=\begin{bmatrix}
1 & 0 \\0 &-1
\end{bmatrix},~
\sigma_4=\begin{bmatrix}
1 & 0 \\ 0&1
\end{bmatrix}.
\end{align*}
 The target matrices are computed as 
$\BX_{k,l}=\frac{1}{r}\BU_{k,l}\BU_{k,l}^*$, where $\BU_{k,l}$ are random $n\times r$ complex orthonormal matrices. We conduct the experiments with $q=6$, $r=1\mbox{ and }3$, and $a$ is selected uniformly at random. When $r=1$, it corresponds to demixing the pure states in quantum  mechanics. The phase transitions of FIHT under the Pauli measurements  is presented in Figure~\ref{fig:phase_pauli}, which shows that FIHT is equally effective for quantum states demixing and $m\approx 3.5 sr(2n-r)$ is sufficient for successful demixing all the tested values of $s$ from $1$ to $7$.

\subsection{An Application from the Internet-of-Things}\label{sec:num_iot}
Demixing problems of the form~\eqref{lowranksum} are  also expected to play an important role in the future
IoT. In this subsection we conduct some numerical simulations for a demixing problem related to the IoT. Before we proceed to the simulations
we describe in more detail how such a demixing problem does arise in the IoT and how it is connected to the demixing of low rank matrices.

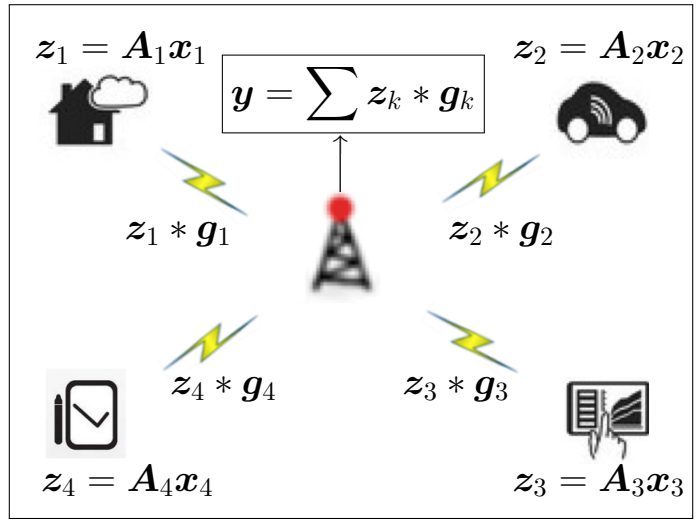
\begin{figure}[!ht]
\begin{center}
\fbox{
\begin{overpic}[width=0.52\textwidth]{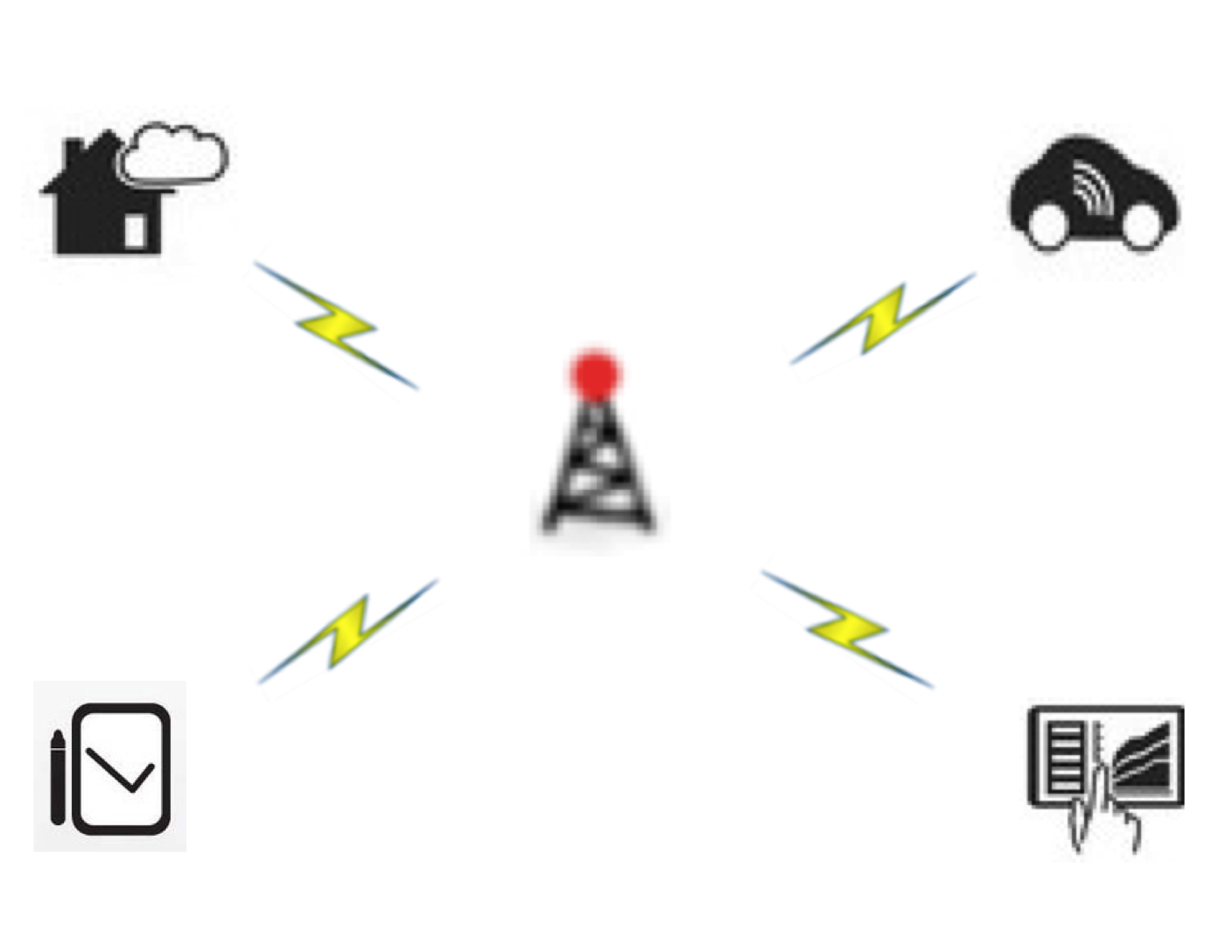}
 \put (1,70) {\large$\displaystyle \bz_1 = \BA_1 \bx_1$}
 \put (15,42) {\large$\displaystyle \bz_1 \ast \bg_1$}
 \put (75,70) {\large$\displaystyle \bz_2 = \BA_2 \bx_2$}
 \put (65,42) {\large$\displaystyle \bz_2 \ast \bg_2$}
 \put (75,3) {\large$\displaystyle \bz_3= \BA_3 \bx_3$}
 \put (58,18) {\large$\displaystyle \bz_3 \ast \bg_3$}
 \put (2,3) {\large$\displaystyle \bz_4 = \BA_4 \bx_4$}
 \put (22,18) {\large$\displaystyle \bz_4 \ast \bg_4$}
 \put (30,63) {\fbox{\large$\displaystyle \by = \sum \bz_k \ast \bg_k$}}
 \put (46.7,52.5) {$\Big\uparrow$}
\end{overpic}
}
\caption{A random access wireless communication scenario in the Internet of Things: Each device communicates with a 
common basestation. No explicit channel estimation is done in order to reduce the scheduling overhead. 
The transmitted signals $\bz_k$ are of the form $\bz_k = \BC_k \bx_k$,
where the $\BC_k$ are encoding matrices and the $\bx_k$ represent the data to be transmitted. 
Each signal $\bz_k$ travels through
a wireless communication channel, represented by $\bg_k$, which acts as convolution. The base station measures the
signal $\by$, which is the sum of all received signals; and it knows the encoding matrices $\BC_k$ of the individual
users, but it does neither know the $\bz_k$'s (or the $\bx_k$'s) nor the transmission channels $\bg_k$. The goal is
to demix the received signal $\by$ and recover the 
signals $\bx_k$ of the individual users.}
\label{fig:wireless}
\end{center}
\end{figure}
In mathematical terms we are dealing---in somewhat simplified form---with the following problem.
We consider a scenario where many different users/devices communicate with a base station,
see~Figure~\ref{fig:wireless} for an illustration of the described setup.
Let $\bx_k\in \C^{n_1}$ be the data that are supposed to be transmitted by the $k$-th user. Before transmission
we encode $\bx_k$ by computing $\bz_k = \BC_k \bx_k$, where $\BC_k\in\C^{m\times n_1}$ is a fixed precoding matrix of the $k$-th user. The encoding matrice $\BC_k$
 differs from user to user but is known to the receiver.  Let $\bg_k$ be the associated 
impulse response of the
communication channel between the $k$-th user and the base station. Since we assume that the channel does not
change during the transmission of the signal $\bz_k$, $\bg_k$ acts as convolution operator, i.e., the signal
arriving at the base station becomes $\bz_k \ast \bg_k$ (ignoring additive noise for simplicity).
However, the channels $\bg_k$ are {\em not known} to the base station (or to the transmitter), since this is 
the overhead information that we want 
to avoid transmitting, as it may change from transmission to transmission.  The signal recorded at the base station
is given by a mixture of all transmitted signals convolved with their respective channel impulse responses, i.e.,
\begin{equation}
\by = \sum_k \bz_k \ast \bg_k.
\label{sumy}
\end{equation}
The goal is to demix the received signal $\by$ and recover the signals $\bx_k$ of the individual users.

But how is~\eqref{sumy} related to our setup in~\eqref{lowranksum}? We first note that in all wireless communication
scenarios of relevance the impulse response $\bg_k$ is compactly supported. In engineering terms the size of the
support of $\bg_k$, denoted here by $n_2$, is called its maximum delay spread. 
In practice,  the length of the transmitted signals $m$ is typically (much) larger than the maximum delay spread $n_2$.
Thus, $\bg_k$ is a $m$-dimensional vector which is formed by the channel impulse
responses being padded with zeroes. We can  write $\bg_k$ as 
$\bg_k = [\bh_k^T , 0,\dots,0]^T$, where $\bh_k$ is the non-zero padded impulse response of length $n_2$. 
Let $\BF$ denote the $m \times m$ unitary Discrete Fourier Transform matrix
and let $\BB$ be the $m\times n_2$ matrix consisting of the first $n_2$ columns of $\BF$.
By the circular convolution theorem, we can express the convolution between $\bz_k$ and $\bg_k$ in the Fourier domain via
\begin{align*}
\frac{1}{\sqrt{m}}\hat{\by}= (\BF \BC_k \bx_k)\odot\BB\bh_k,\numberthis\label{eq:iot_model}
\end{align*}
where $\odot$ denotes the componentwise product.
Let   $\bc_{k,p}^*$ represent the $p$-th row of $\BF \BC_k$ and let $\bb_{p}^*$ represent the $p$-th row of $\BB$.
A simple calculation shows that~\cite{LS15}
\begin{equation}
\label{rankone}
 [(\BF \BC_k \bx_k)\odot\BB\bh_k]_p= (\bc_{k,p}^*\bx_k)\cdot(\bb_p^*\bh_k)=\la\bc_{k,p}\bar{\bb}_p^*,\bx_k\bar{\bh}_k^*\ra
=\la\bc_{k,p}\bar{\bb}_p^*,\BX_k\ra,
\end{equation}
where $\BX_k: = \bx_k\bar{\bh}_k^* $ is an $n_1\times n_2$ rank-one matrix containing the unknown signal and the unknown channel of
the $k$-th user,
while  $\bc _{k,p}$ and $\bb_p^{\ast} $ are known vectors (since $\BC_k$ and $\BB$ are known). 
Hence, at the base station the observed signal $\by = \sum_{k=1}^s \bz_k \ast \bg_k$ can be expressed 
in the Fourier domain as  $\frac{1}{\sqrt{m}}\hat{\by} = \sum_{k=1}^s \A_k(\BX_k)$, which is of the same form as~\eqref{lowranksum}
(modulo replacing $\frac{1}{\sqrt{m}}\hat{\by}$ with $\by$). The measurement matrices in this scenario are  given by $\BA_{k,p} =\bc_{k,p}\bar{\bb}_p^*$.

\begin{figure}[!ht]
\centering
\subfloat[]{\includegraphics[width=0.45\textwidth]{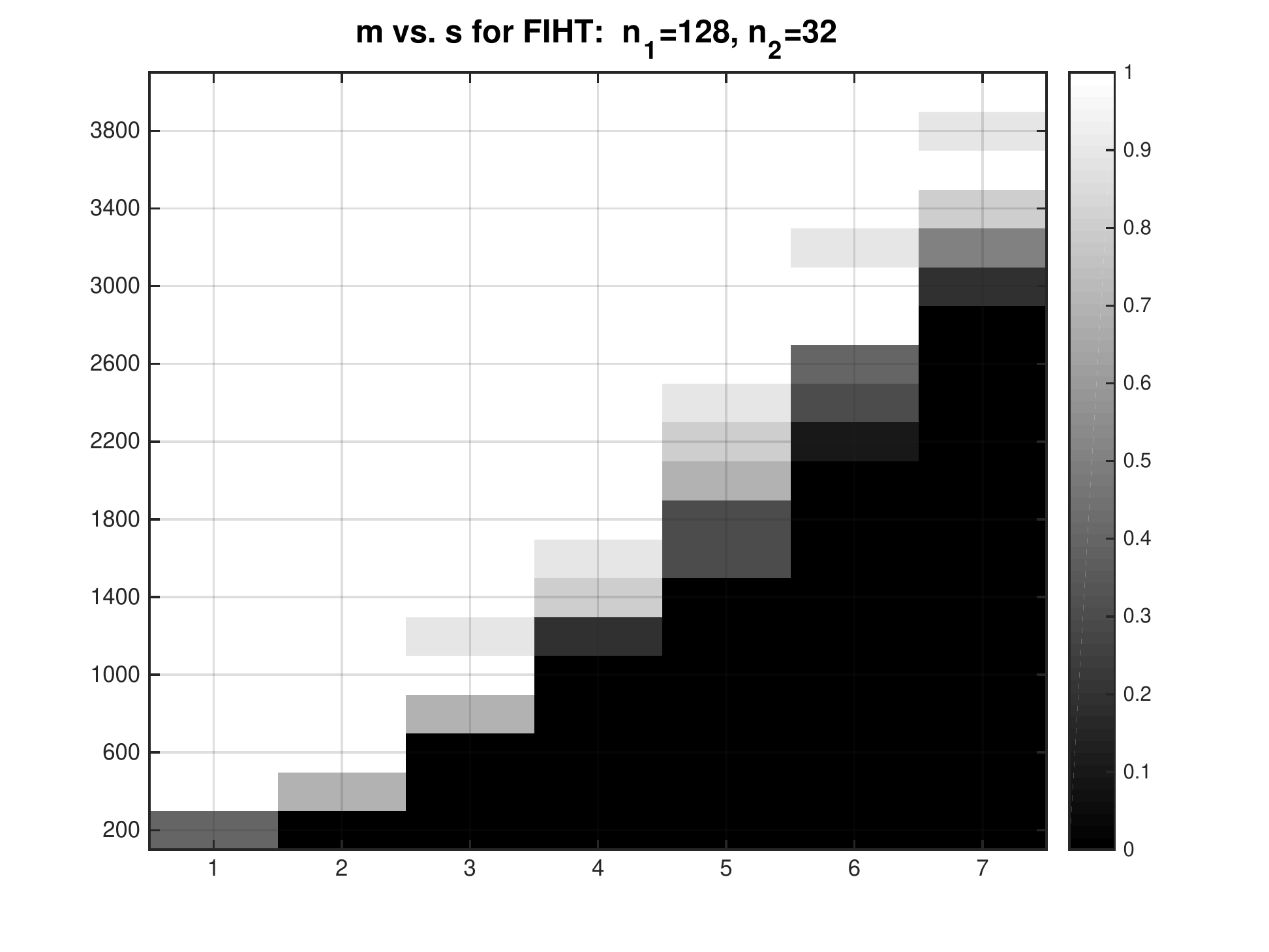}\label{fig:FIHT_Fourier_Gaussian}}
\subfloat[]{\includegraphics[width=0.45\textwidth]{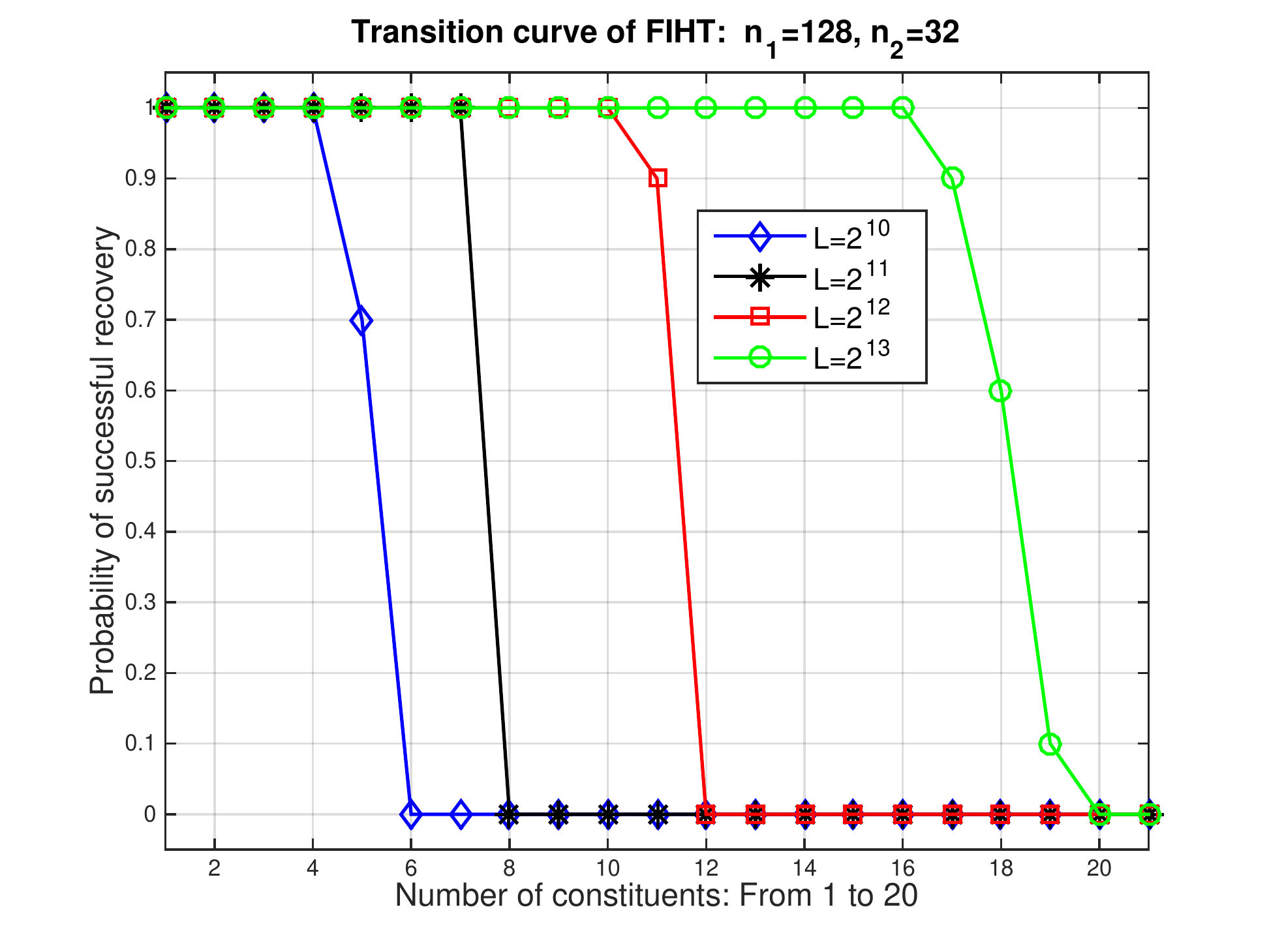}\label{fig:FIHT_Fourier_Hadamard}}
\caption{Empirical phase transitions of FIHT for the application from the IoT.}
\label{fig:phase_fourier}
\end{figure}

We test FIHT  (Algorithm~\ref{alg:fiht}) for the demixing problem described in \eqref{eq:iot_model} with $n_1=128$ and $n_2=32$, which are of interest in practice. As a result, the constituents are non-square matrices. Two types of encoding matrices are tested: a) the entries of $\BC_k$ are standard complex Gaussian variables; and b) $\BC_k=\BD_k\BH$, where $\BH$ is the first $n_2$ columns of an $m\times m$ Hadamard matrix, and then premultiplied by an $m\times m$ diagonal random sign matrix $\BD$. The target signals $\bx_k$ and $\bh_k$ are complex Gaussian random vectors.
The phase transition plot for $\BC_k$ being a Gaussian matrix is presented in Figure~\ref{fig:FIHT_Fourier_Gaussian}, where the number of constituents varies from $1$ to $7$. Figure~\ref{fig:FIHT_Fourier_Gaussian} displays a linear correlation between the number of necessary measurements and the number of the constituents, and it  shows that $m\approx 3s(n_1+n_2)$ number of measurements are sufficient for successful demixing of the transmission signals and the impulse responses. When $\BC_k$ is formed from a partial Hadamard matrix, we only test $m=2^{10}, 2^{11},~2^{12}$, and $2^{13}$ due to the Hadamard conjecture \cite{HeWa1978Hadamard}, but a larger range of $s$ are tested. The phase transition plot presented in Figure~\ref{fig:FIHT_Fourier_Hadamard} shows that FIHT is still able to successfully demix the transmission signals and the impulse responses when $m\approx 3s(n_1+n_2)$.

\subsection{Robustness to Additive Noise}
We explore the robustness of FIHT against additive noise by conducting tests on the problem instances set up in the last three subsections. We have established in Theorem~\ref{thm:iht_noise} that Algorithm~\ref{alg:iht} (IHT) is stable in the presence of additive noise, but could not provide a similar theoretical robustness guarantee for FIHT because of the more involved nature of the current  convergence analysis framework of FIHT. Yet, empirical evidence clearly indicates that FIHT also shows the desirable stability vis-\`a-vis noise, as we will demonstrate in this subsection.

We assume that the measurement vectors $\by$ in the tests are corrupted by 
\begin{align*}
\be=\sigma\cdot\ln\by\rn\cdot\frac{\bw}{\ln\bw\rn},
\end{align*}
where $\bw$ is an $m\times 1$ standard Gaussian vector, either real or complex up to the testing environment, and $\sigma$ takes  nine different values of $\sigma$ from $10^{-4}$ to $1$. 
For each problem instance, two values of the number of measurements are tested. The average relative reconstruction error out in dB against the signal to noise ratio (SNR)  is plotted in Figure~\ref{fig:phase_robust}. The figure shows that the relative reconstruction error scales linearly with the noise levels under different measurement schemes. Moreover,  as desired, the relative reconstruction error decreases linearly on  a log-log scale as the number of measurements increases.
\begin{figure}[!ht]
\centering
\subfloat[]{\includegraphics[width=0.45\textwidth]{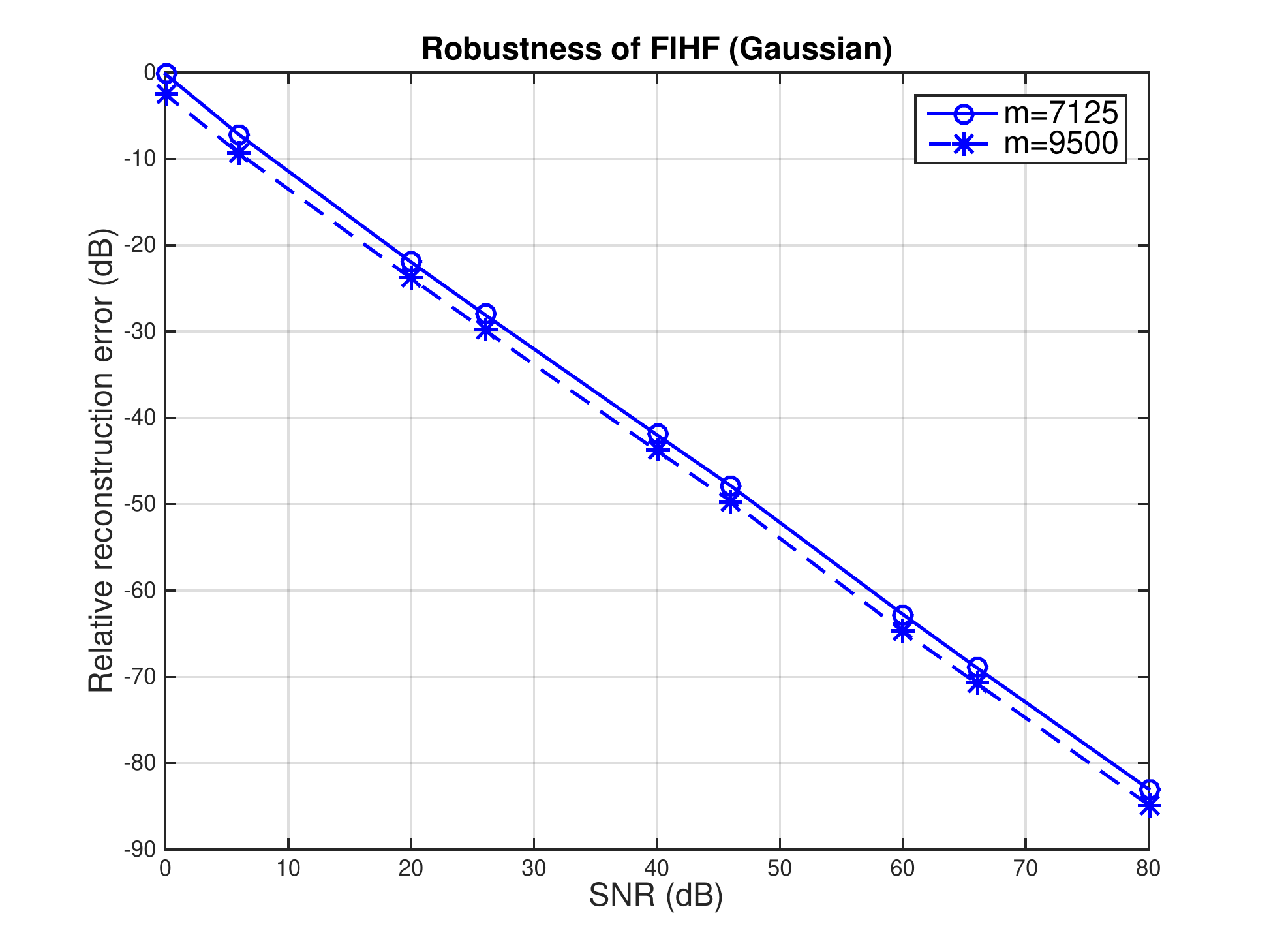}\label{fig:rob_gaussian}}\\
\subfloat[]{\includegraphics[width=0.45\textwidth]{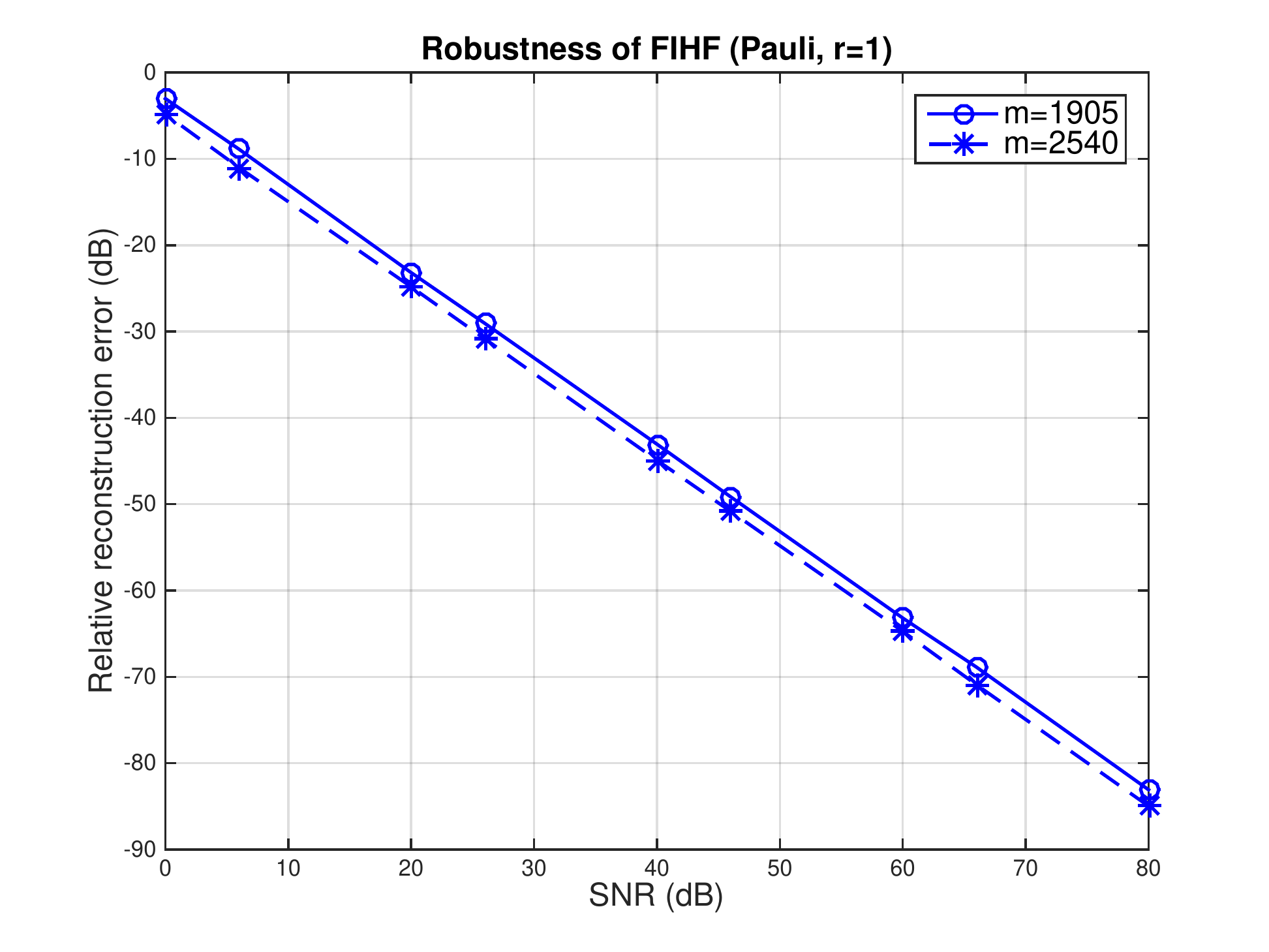}\label{fig:rob_pauli1}}
\subfloat[]{\includegraphics[width=0.45\textwidth]{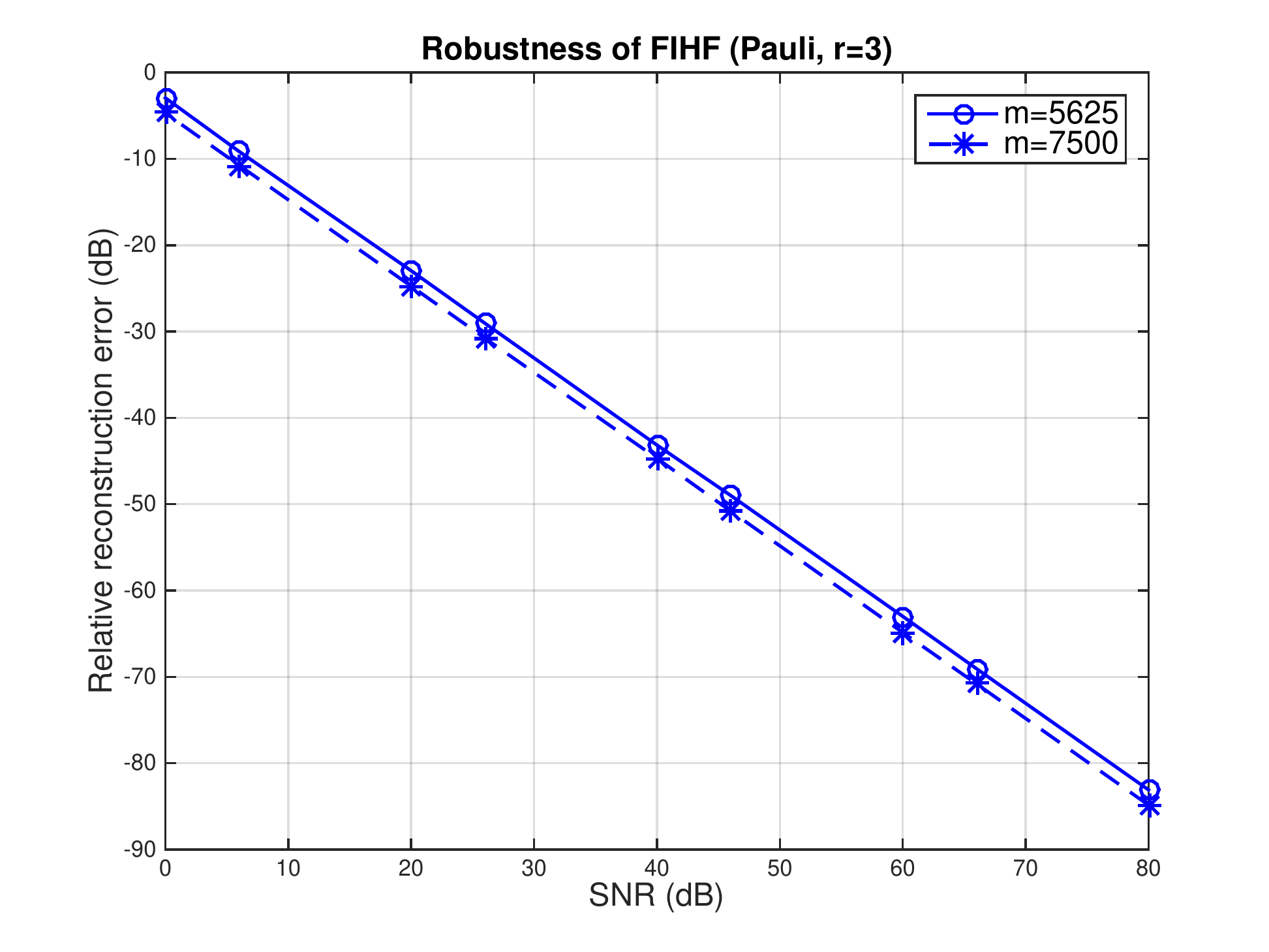}\label{fig:rob_pauli3}}\\
\subfloat[]{\includegraphics[width=0.45\textwidth]{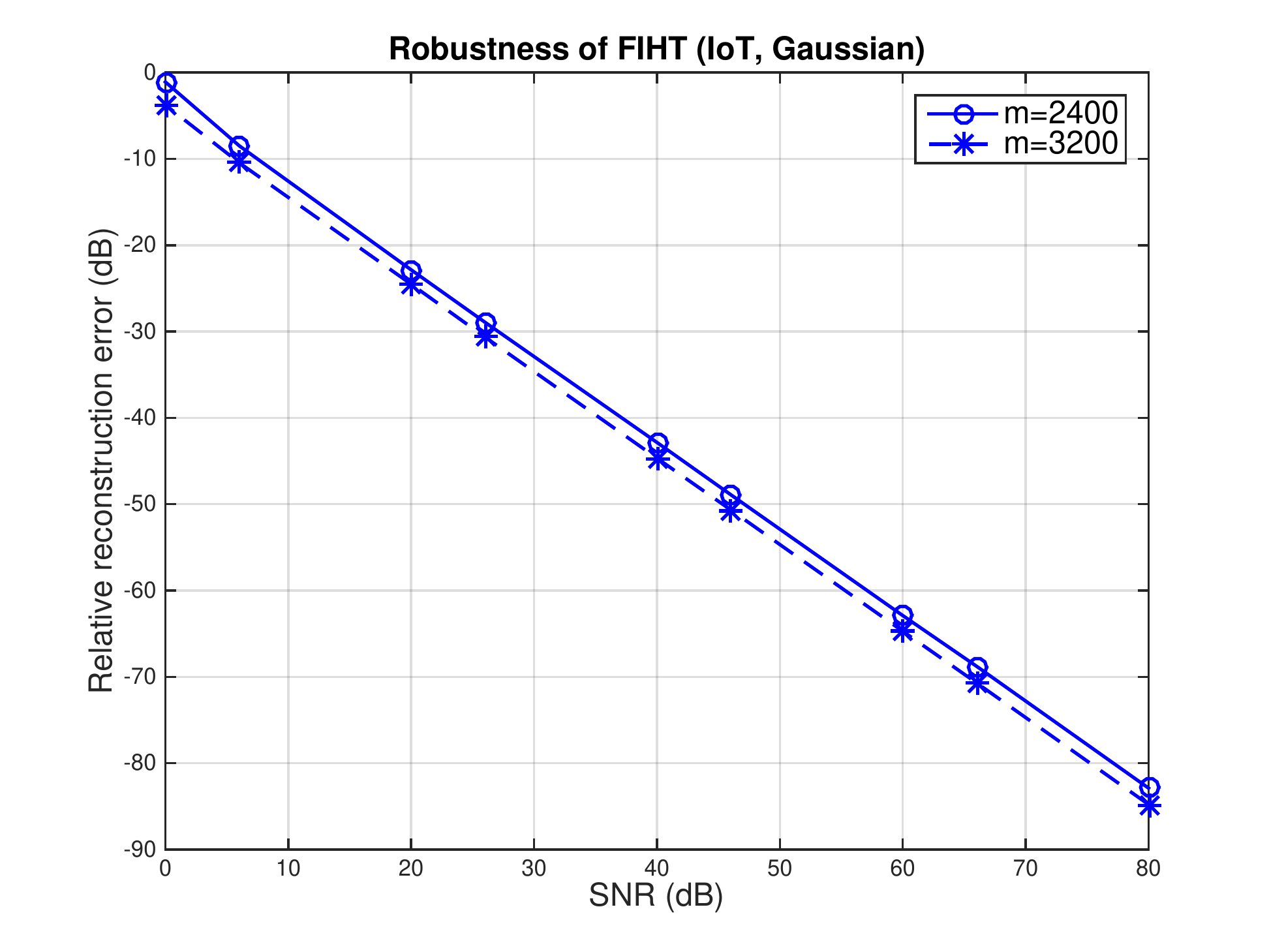}\label{fig:rob_fou_gaussian}}
\subfloat[]{\includegraphics[width=0.45\textwidth]{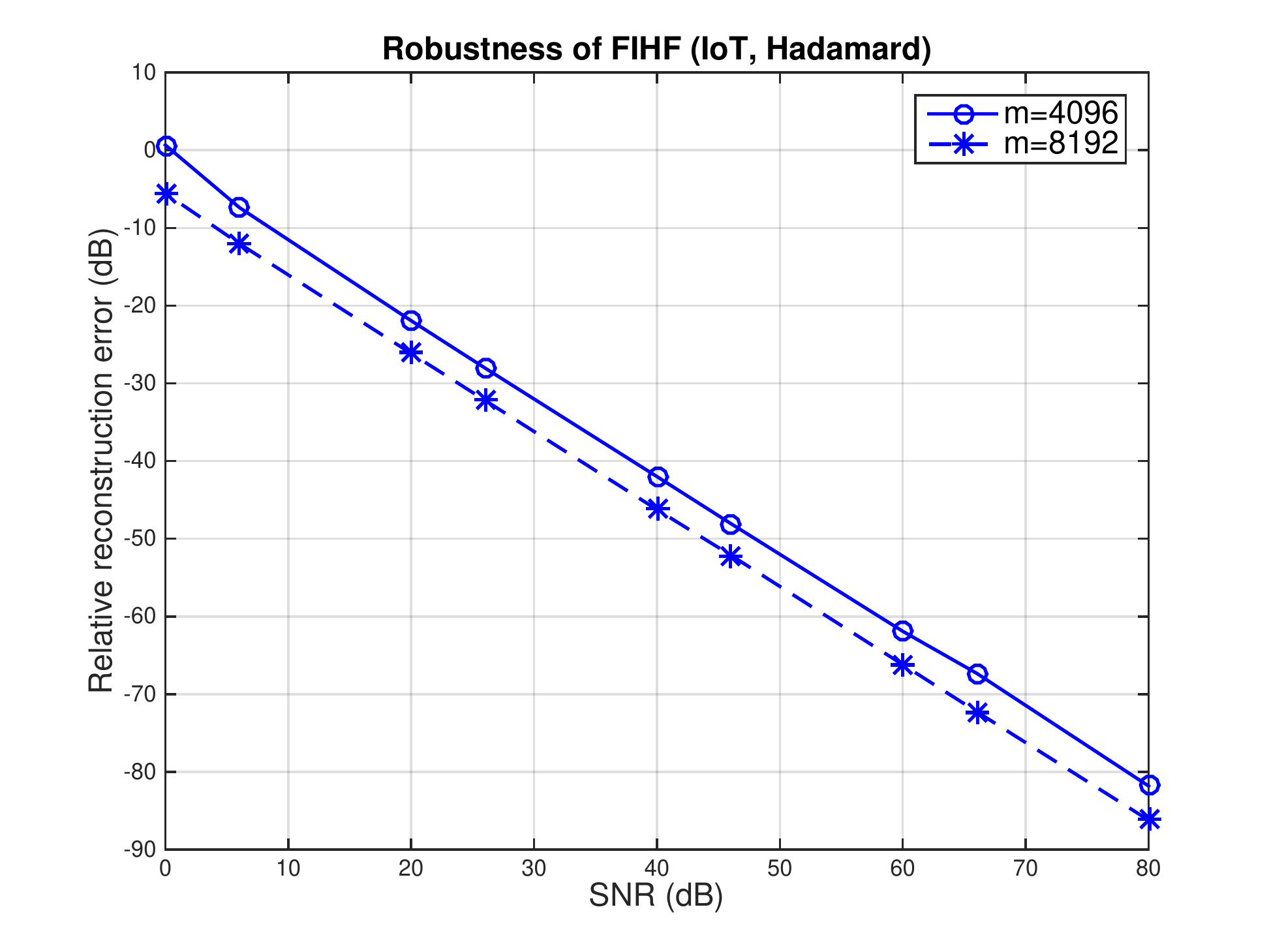}\label{fig:rob_fou_hadamard}}
\caption{Robustness of FIHT under different measurements: (a) Gaussian measurements with $s=5$, (b) and (c) Pauli measurements with $s=5$ , (d) and (e) applications in the IoT with $s=5$ and $s=10$, respectively.  }
\label{fig:phase_robust}
\end{figure}

\subsection{A Rank Increasing Heuristic}
\begin{figure}[!ht]
\centering
\subfloat[]{\includegraphics[width=0.45\textwidth]{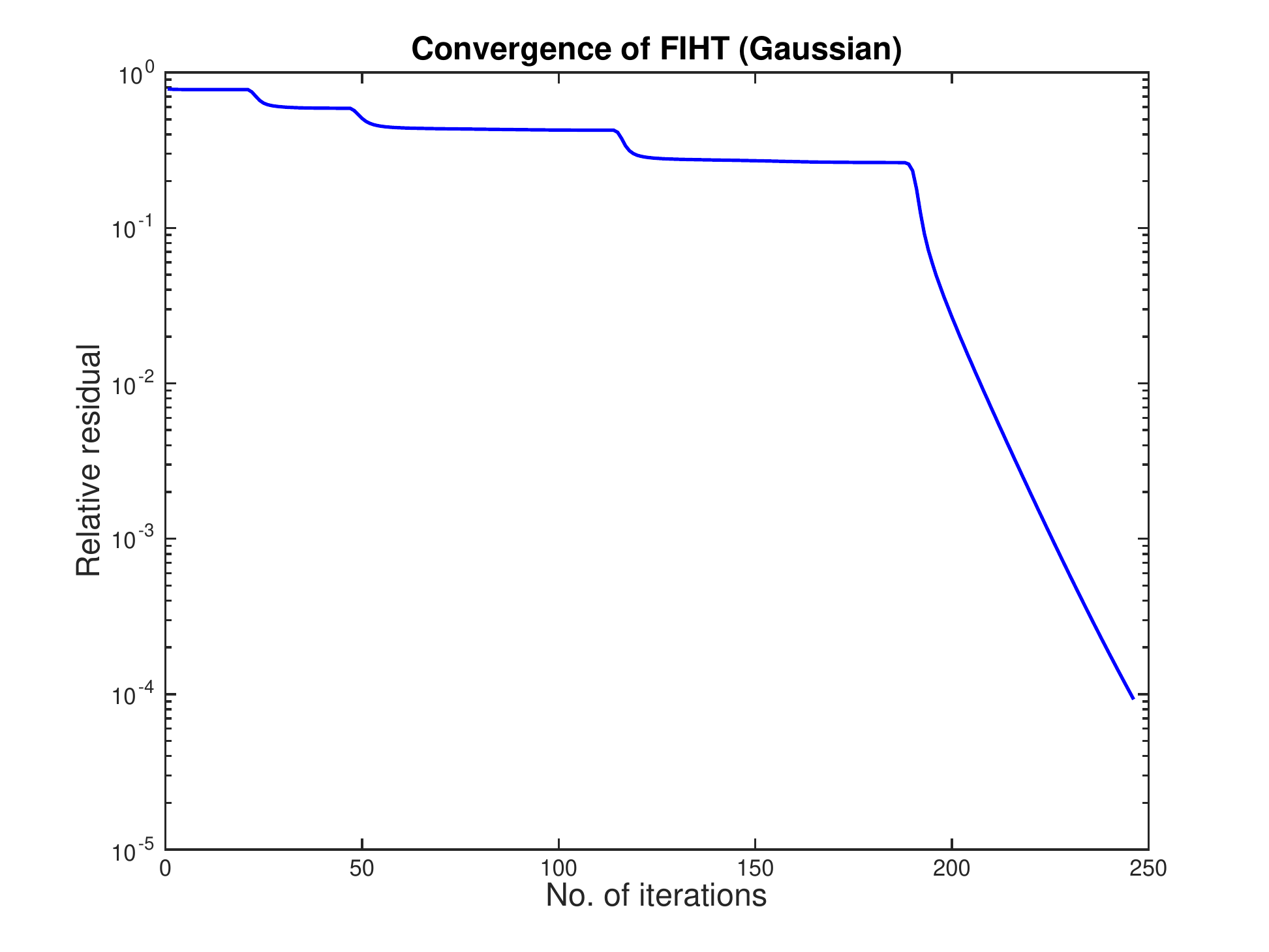}\label{fig:con_gaussian}}
\subfloat[]{\includegraphics[width=0.45\textwidth]{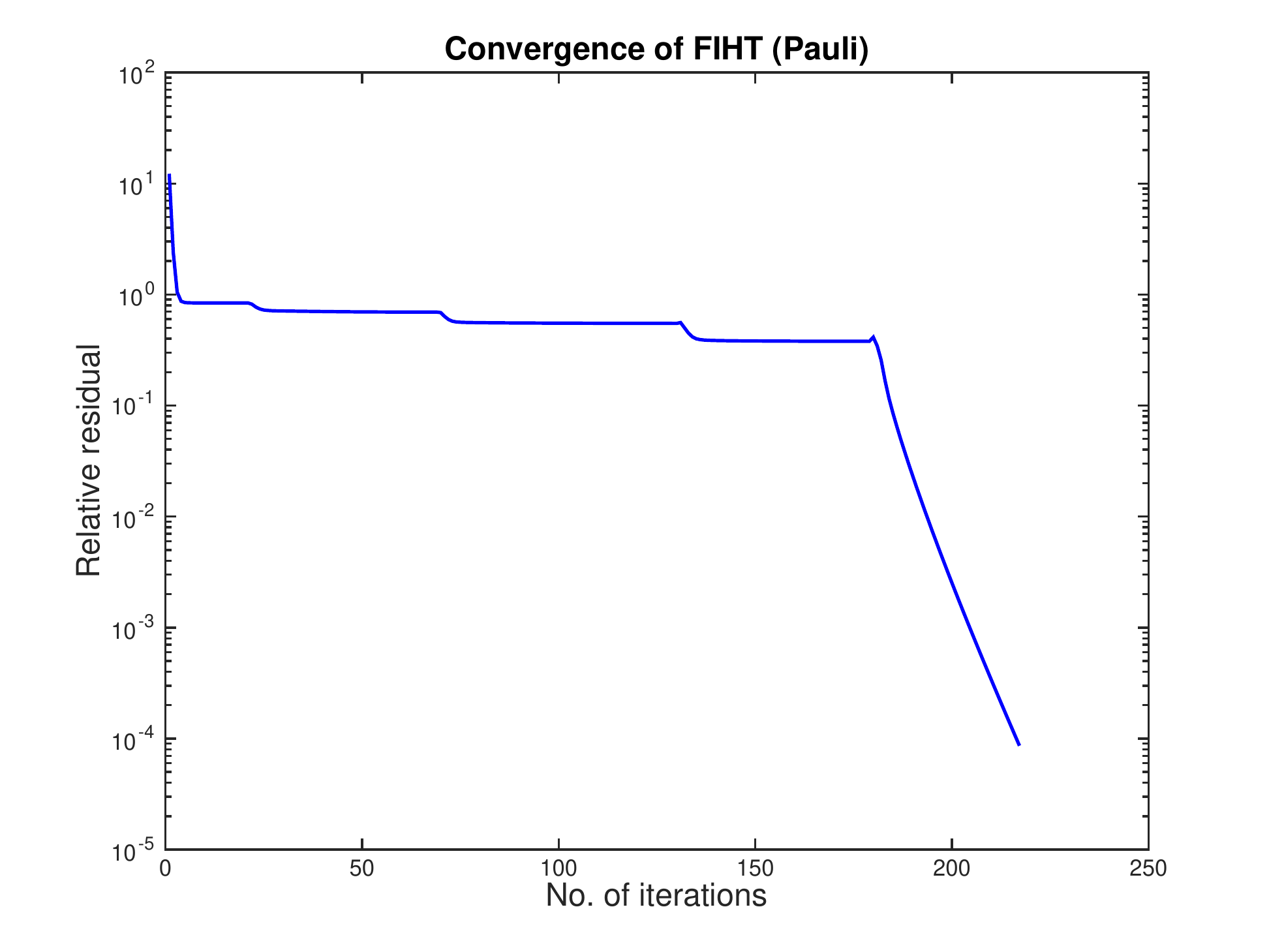}\label{fig:con_pauli}}
\caption{Rank increasing heuristic for FIHT: (a) Gaussian measurements with $s=5$, $n=50$, $r=5$, and $m=3sr(2n-r)$, and (b) Pauli measurements  with $s=5$, $n=2^{6}$, $r=5$, and $m=3sr(2n-r)$. The jumps in the two error curves  correspond to the iteration steps when the algorithm increases the assumed rank of the unknown matrices.}
\label{fig:con_fiht}
\end{figure}

While in some applications the rank of the target matrices is known a priori, for example in  pure states demxing  and the IoT application, 
we may not know the matrix rank exactly in other applications. Here, we suggest a rank increasing heuristic when the rank of the matrices is not known. Starting from rank one, we execute FIHT until the residual is not decreasing significantly; and  then we successively increase the rank and repeat this procedure until the residual is sufficiently small. We test this heuristic using problems instances from Subsections~\ref{sec:num_gaussian} and \ref{sec:num_pauli}.  The relative residual plotted against the number of iterations is presented in Figure~\ref{fig:con_fiht}.
The jumps in the error curves correspond to the iteration steps when the algorithm increases the assumed rank of the unknown matrices by 1.

\section{Proofs}\label{s:proofs}
\subsection{Proof of Theorem~\ref{thm:ARIP}}\label{ss:arip}
\begin{proof}
Let $\{\BZ_k\}_{k=1}^s$ be a set of fixed matrices with $\rank(\BZ_k)\leq r$ for all $1\leq k\leq s$. One can easily see that 
\begin{align*}
\sqrt{m}\sum_{k=1}^s\A_k(\BZ_k) =\begin{bmatrix}
\sum_{k=1}^s\la\BA_{k,1},\BZ_k\ra\\
\vdots\\
\sum_{k=1}^s\la\BA_{k,m},\BZ_k\ra
\end{bmatrix}.
\end{align*}
Since for any $1\leq p\leq m$, $\BA_{1,p}, \BA_{2,p}, \cdots,\BA_{s,p}$ are independent from each other, one has 
\[
\sum_{k=1}^s\la\BA_{k,p},\BZ_k\ra\sim \N(0,\sum_{k=1}^s\ln\BZ_k\rn_F^2).
\]
Therefore,
\begin{align*}
m\ln\sum_{k=1}^s\A_k(\BZ_k)\rn^2=\sum_{p=1}^m \lab \sum_{k=1}^s\la\BA_{k,p},\BZ_k\ra\rab^2\sim \sum_{p=1}^m \lb\sum_{k=1}^s\ln \BZ_k\rn_F^2\rb\xi_p^2,
\end{align*}
where $\xi_p^2,~p=1,\cdots,m$ are i.i.d $\chi^2$ variables. Consequently, $$\E\lsb m\ln\sum_{k=1}^s\A_k(\BZ_k)\rn^2\rsb = m \sum_{k=1}^s\ln \BZ_k\rn_F^2,$$ and the application of the Bernstein inequality for Chi-squared variables \cite{Ver2011rand} gives 
\begin{align*}
&\Pr\lb \lab m\ln\sum_{k=1}^s\A_k(\BZ_k)\rn^2-m \sum_{k=1}^s\ln \BZ_k\rn_F^2\rab\geq t\rb\\
&\leq 2\exp\lb-\min\lcb\frac{t^2}{8m\lb\sum_{k=1}^s\ln \BZ_k\rn_F^2\rb^2},\frac{t}{8\sum_{k=1}^s\ln \BZ_k\rn_F^2}\rcb\rb.
\end{align*}
Taking $t=\delta\cdot m\sum_{k=1}^s\ln\BZ_k\rn_F^2$ gives 
\begin{align*}
\Pr\lb \lab \ln\sum_{k=1}^s\A_k(\BZ_k)\rn^2- \sum_{k=1}^s\ln \BZ_k\rn_F^2\rab\geq \delta\sum_{k=1}^s\ln\BZ_k\rn_F^2\rb
\leq 2\exp\lb-\frac{m\delta^2}{8}\rb
\end{align*}
for $0<\delta<1$.  This implies for fixed $\lcb\BZ_k\rcb_{k=1}^s$, we have that
\begin{align*}
(1-\delta)\sum_{k=1}^s\ln\BZ_k\rn_F^2\leq \ln\sum_{k=1}^s\A_k(\BZ_k)\rn^2\leq (1+\delta)\sum_{k=1}^s\ln\BZ_k\rn_F^2\numberthis\label{eq:ARIP_fixed}
\end{align*}
holds with probability at least $1-2e^{-m\delta^2/8}$.

The rest of the proof  follows essentially from the proof of Theorem~{2.3} in \cite{candesplan2009oracle}.
Due to the homogeneity of the problem with respect to $\sum_{k=1}^s\ln\BZ_k\rn_F^2$,  we only need to show \eqref{eq:ARIP_fixed} for the set of matrices in
\begin{align*}
\S_r=\lcb[\BZ_1,\cdots,\BZ_s]^T:~\rank(\BZ_k)\leq r\mbox{ and }\sum_{k=1}^s\ln\BZ\rn_F^2=1\rcb.
\end{align*}
Based on the SVD of each $\BZ_k$,  we have
\begin{align*}
\S_r=\lcb [\BU_1\BS_1\BV_1^*,\cdots,\BU_s\BS_s\BV_s^*]^T:~\BU_k^*\BU_k=\BI,\BV_k^*\BV=\BI \mbox{ and }\sum_{k=1}^s\sum_{i=1}^r|\BS_k^{i,i}|^2=1\rcb.
\end{align*}
Define $\O_{n,r} =  \{\BQ\in\R^{n\times r}:~\BQ^*\BQ=\BI\}$. It has been shown in \cite{candesplan2009oracle} that $\O_{n,r}$ has a $\epsilon$-net $\overline{\O}_{n,r}\subset \O_{n,r}$ of cardinality $|\overline{\O}_{n,r}|\leq(3/\epsilon)^{nr}$ under the $\ln\cdot\rn_{1,2}$ norm. In other words, for any $\BQ\in \O_{n,r}$, there exists a $\overline{\BQ}\in\overline{\O}_{n,r}$ such that 
\begin{align*}
\ln\BQ-\overline{\BQ}\rn_{1,2}=\max_{1\leq i\leq r}\ln \BQ(:,i)-\overline{\BQ}(:,i)\rn\leq \epsilon.
\end{align*}
Let $\D$ be the set of $rs\times rs$ diagonal matrices  with unit norm diagonal entries. It is well known that $\D$ has a $\epsilon$-net $\overline{\D}\subset\D$ of cardinality $|\overline{\D}|\leq (3/\epsilon)^{rs}$ \cite{Ver2011rand}.
Define 
\begin{align*}{
\overline{\S}_r = \lcb [\overline{\BU}_1\overline{\BS}_1\overline{\BV}_1^*,\cdots,\overline{\BU}_s\overline{\BS}_s\overline{\BV}_s^*]^T:~\overline{\BU}_k\in\overline{\O}_{n,r},~\overline{\BV}_k\in\overline{\O}_{n,r},\mbox{ and }\diag(\overline{\BS}_1,\cdots, \overline{\BS}_k)\in\overline{\D}\rcb}.
\end{align*}
We are going to show that $\overline{\S}_r$ is a $3\epsilon$-net of $\S_r$ with cardinality $|\overline{\S}_r|\leq (3/\epsilon)^{(2n+1)rs}$. 

First note that for any $[\BU_1\BS_1\BV_1^*,\cdots,\BU_s\BS_s\BV_s^*]^T$, there exists a $[\overline{\BU}_1\overline{\BS}_1\overline{\BV}_1^*,\cdots,\overline{\BU}_s\overline{\BS}_s\overline{\BV}_s^*]^T\in\overline{\S}_r$ such that 
$\ln \BU_k-\overline{\BU}_k\rn_{1,2}\leq\epsilon$,  $\ln \BV_k-\overline{\BV}_k\rn_{1,2}\leq\epsilon$, and $\sqrt{\sum_{k=1}^s\ln\BS_k-\overline{\BS}_k\rn_F^2}\leq\epsilon$. Thus, 
\begin{align*}
\ln\begin{bmatrix}{\BU}_1{\BS}_1{\BV}_1^*\\
\vdots\\
{\BU}_s{\BS}_s{\BV}_s^*
\end{bmatrix}
-
\begin{bmatrix}
\overline{\BU}_1\overline{\BS}_1\overline{\BV}_1^*\\
\vdots\\
\overline{\BU}_s\overline{\BS}_s\overline{\BV}_s^*
\end{bmatrix}\rn_F&\leq
\ln\begin{bmatrix}
(\BU_1-\overline{\BU}_1)\BS_1\BV_1^*\\
\vdots\\
(\BU_s-\overline{\BU}_s)\BS_s\BV_s^*
\end{bmatrix}\rn_F+\ln
\begin{bmatrix}
\overline{\BU}_1(\BS_1-\overline{\BS}_1)\BV_1^*\\
\vdots\\
\overline{\BU}_s(\BS_s-\overline{\BS}_s)\BV_s^*
\end{bmatrix}\rn_F\\
&\quad+\ln
\begin{bmatrix}
\overline{\BU_1}\overline{\BS}_1(\BV_1-\overline{\BV}_1)^*\\
\vdots\\
\overline{\BU_s}\overline{\BS}_s(\BV_s-\overline{\BV}_s)^*
\end{bmatrix}\rn_F\\
&:=\II_1+\II_2+\II_3.
\end{align*}
Since
\begin{align*}
\ln\begin{bmatrix}
(\BU_1-\overline{\BU}_1)\BS_1\BV_1^*\\
\vdots\\
(\BU_s-\overline{\BU}_s)\BS_s\BV_s^*
\end{bmatrix}\rn_F^2 &= \sum_{k=1}^s\ln(\BU_k-\overline{\BU}_k)\BS_k\BV_k^*\rn_F^2\\
&=\sum_{k=1}^s\ln (\BU_k-\overline{\BU}_k)\BS_k\rn_F^2\\
&=\sum_{k=1}^s\sum_{i=1}^r|\BS_k^{i,i}|^2\ln\BU_k(:,i)-\overline{\BU}_k(:,i)\rn^2\\
&\leq\sum_{k=1}^s\sum_{i=1}^r|\BS_k^{i,i}|^2\ln \BU_k-\overline{\BU}_k\rn_{1,2}^2\\
&\leq \epsilon^2 \sum_{k=1}^s\sum_{i=1}^r|\BS_k^{i,i}|^2=\epsilon^2,
\end{align*}
we have $\II_1\leq \epsilon$. Similarly, the third term can be bounded as $\II_3\leq \epsilon$. To bound $\II_2$, note that 
\begin{align*}
\ln
\begin{bmatrix}
\overline{\BU}_1(\BS_1-\overline{\BS}_1)\BV_1^*\\
\vdots\\
\overline{\BU}_s(\BS_1-\overline{\BS}_s)\BV_s^*
\end{bmatrix}\rn_F^2&=\sum_{k=1}^s\ln \overline{\BU}_k(\BS_k-\overline{\BS}_k)\BV_k^*\rn_F^2\\
&=\sum_{k=1}^s\ln \BS_k-\overline{\BS}_k\rn_F^2=\epsilon^2.
\end{align*}
So we also have $\II_2\leq\epsilon$. Combining the bounds for $\II_1,~\II_2$ and $\II_3$ together implies that $\overline{\S}_r$ is a $3\epsilon$-net of $\S_r$ with cardinality $|\overline{\S}_r|\leq (3/\epsilon)^{(2n+1)rs}$. Therefore, 
\begin{align*}
&\Pr\lb\max_{\{\BZ_k\}_{k=1}^s\in\overline{\S}_r} \lab\ln\sum_{k=1}^s\A_k(\BZ_k)\rn^2- \sum_{k=1}^s\ln \BZ_k\rn_F^2\rab\geq \delta/2\rb
\leq 2\lb\frac{3}{\epsilon}\rb^{(2n+1)rs}e^{-m\delta^2/32}
\leq 2e^{-m\delta^2/64}
\end{align*}
provided
\begin{align*}
m\geq 64\delta^{-2}(2n+1)rs\log(3/\epsilon).
\end{align*}
So  for all  $\lcb\overline{\BZ}_k\rcb\in\overline{\S}_k$, we have 
\begin{align*}
1-\delta/2\leq\ln\sum_{k=1}^s\A_k(\overline{\BZ}_k)\rn\leq 1+ \delta/2 
\end{align*}
with probability at least $1-2e^{-m\delta^2/64}$.

Define
\begin{align*}
\kappa_r = \sup_{\{\BZ_k\}_{k=1}^s\in{\S_r}}\ln\sum_{k=1}^s\A_k(\BZ_k)\rn.
\end{align*}In the following, we will take $\epsilon=\delta/(12\sqrt{2})$. So $\overline{\S}_r$ is  a $\delta/(4\sqrt{2})$-net of $\S_r$;
and for any $\{\BZ_k\}_{k=1}^s\in\S_r$, there exist  $\{\overline{\BZ}_k\}_{k=1}^s\in\overline{\S}_r$ such that 
\begin{align*}\ln
\begin{bmatrix}
\BZ_1-\overline{\BZ}_1\\
\vdots\\
\BZ_s-\overline{\BZ}_s
\end{bmatrix}\rn_F=\sqrt{\sum_{k=1}^s\ln \BZ_k-\overline{\BZ}_k\rn_F^2}\leq \delta/(4\sqrt{2}).
\end{align*}
Thus, 
\begin{align*}
\ln\sum_{k=1}^s\A_k(\BZ_k)\rn & \leq \ln \sum_{k=1}^s\A_k\lb\overline{\BZ}_k\rb\rn+\ln\sum_{k=1}^s\A_k\lb\BZ_k-\overline{\BZ}_k\rb\rn\\
&\leq 1+\delta/2+\ln\sum_{k=1}^s\A_k\lb\BZ_k-\overline{\BZ}_k\rb\rn.\numberthis\label{eq:rip_01}
\end{align*}
Since $\BZ_k-\overline{\BZ}_k$ is a matrix of rank at most  $2r$, we can decompose it as 
\begin{align*}
\BZ_k-\overline{\BZ}_k = \BY_k^1+\BY_k^2,
\end{align*}
where $\la\BY_k^1,\BY_k^2\ra=0$, $\rank(\BY_k^i)\leq r ~(i=1,2)$, and $\ln\BY_k^1\rn_F^2+\ln\BY_k^2\rn_F^2=\ln \BZ_k-\overline{\BZ}_k \rn_F^2$. Then it follows that
\begin{align*}
\ln\sum_{k=1}^s\A_k\lb\BZ_k-\overline{\BZ}_k\rb\rn&\leq \ln\sum_{k=1}^s\A_k(\BY_k^1)\rn+\ln\sum_{k=1}^s\A_k(\BY_k^2)\rn\\
&\leq \kappa_r\lb\sqrt{\sum_{k=1}^s\ln\BY_k^1\rn_F^2}+\sqrt{\sum_{k=1}^s\ln\BY_k^2\rn_F^2}\rb\\
&\leq\sqrt{2} \kappa_r\sqrt{\sum_{k=1}^s\ln\BY_k^1\rn_F^2+\sum_{k=1}^s\ln\BY_k^2\rn_F^2}\\
&=\sqrt{2} \kappa_r\sqrt{\sum_{k=1}^s\ln \BZ_k-\overline{\BZ}_k \rn_F^2}\\
&\leq \frac{\delta\kappa_r}{4}.\numberthis\label{eq:rip_02}
\end{align*}
where the second line follows from $\lcb \BY_k^i/\sqrt{\sum_{k=1}^s\ln\BY_k^i\rn_F^2}\rcb_{k=1}^s\in\S_r$ for $i=1,2$.  Inserting \eqref{eq:rip_02} into \eqref{eq:rip_01} gives 
$$
\kappa_r\leq 1+\delta/2+\delta\kappa_r/4,
$$
which implies $\kappa_r\leq 1+\delta$ and 
\begin{align*}
\ln\sum_{k=1}^s\A_k(\BZ_k)\rn\leq 1+\delta\numberthis\label{eq:rip_03}
\end{align*}
for all $\{\BZ_k\}_{k=1}^s\in{\S_r}$. The lower bound can be obtained in the following way: 
\begin{align*}
\ln\sum_{k=1}^s\A_k(\BZ_k)\rn&\geq \ln \sum_{k=1}^s\A_k\lb\overline{\BZ}_k\rb\rn-\ln\sum_{k=1}^s\A_k\lb\BZ_k-\overline{\BZ}_k\rb\rn\\
&\geq 1-\delta/2-\delta(1+\delta)/4\\
&\geq 1-\delta.\numberthis\label{eq:rip_04}
\end{align*}
where the second line follows from \eqref{eq:rip_02} and the bound for $\kappa_r$. The upper and lower bounds in \eqref{eq:rip_03} and \eqref{eq:rip_04} can be easily transferred into the desired ARIP bounds by a change of variables.
\end{proof}
\subsection{Proof of Theorem~\ref{thm:iht}}\label{ss:iht}
\begin{proof}
Let $\BU_{k,l+1}$, $\BU_{k,l}$, and $\BU_k$ be the left singular vectors of $\BX_{k,l+1}$, $\BX_{k,l}$ and $\BX_k$, respectively. Let $\BQ_{k,l}$ be an $n\times 3r$ orthonormal matrix which spans the union of the column spaces of $\BU_{k,l+1}$, $\BU_{k,l}$, and $\BU_k$. We have $\BX_{k,l+1}=\P_{\BU_{k,l+1}}(\BX_{k,l}+\alpha_{k,l}\BG_{k,l})$.  Moreover, by comparing the following two equalities,
\begin{align*}
\ln \BX_{k,l+1}-\lb\BX_{k,l}+\alpha_{k,l}\BG_{k,l}\rb\rn_F^2&=\ln\BX_{k,l+1}-\P_{\BQ_{k,l}}(\BX_{k,l}+\alpha_{k,l}\BG_{k,l})\rn_F^2+\ln\P_{\BQ_{k,l}^\perp}(\BX_{k,l}+\alpha_{k,l}\BG_{k,l})\rn_F^2,\\
\ln \BX_{k}-\lb\BX_{k,l}+\alpha_{k,l}\BG_{k,l}\rb\rn_F^2&=\ln\BX_{k}-\P_{\BQ_{k,l}}(\BX_{k,l}+\alpha_{k,l}\BG_{k,l})\rn_F^2+\ln\P_{\BQ_{k,l}^\perp}(\BX_{k,l}+\alpha_{k,l}\BG_{k,l})\rn_F^2,
\end{align*}
we have 
\begin{align*}
\ln\BX_{k,l+1}-\P_{\BQ_{k,l}}(\BX_{k,l}+\alpha_{k,l}\BG_{k,l})\rn_F\leq \ln\BX_{k}-\P_{\BQ_{k,l}}(\BX_{k,l}+\alpha_{k,l}\BG_{k,l})\rn_F
\end{align*}
since $\BX_{k,l+1}$ is the best rank $r$ approximation of $\BX_{k,l}+\alpha_{k,l}\BG_{k,l}$. Thus, 
\begin{align*}
\ln \BX_{k,l+1}-\BX_k\rn_F&\leq \ln\BX_{k,l+1}-\P_{\BQ_{k,l}}(\BX_{k,l}+\alpha_{k,l}\BG_{k,l})\rn_F+\ln\BX_{k}-\P_{\BQ_{k,l}}(\BX_{k,l}+\alpha_{k,l}\BG_{k,l})\rn_F\\
&\leq 2\ln\BX_{k}-\P_{\BQ_{k,l}}(\BX_{k,l}+\alpha_{k,l}\BG_{k,l})\rn_F\\
&=2\ln\P_{\BQ_{k,l}}(\BX_{k,l})-\P_{\BQ_{k,l}}(\BX_{k})+\alpha_{k,l}\P_{\BQ_{k,l}}(\BG_{k,l})\rn_F.
\end{align*}
Consequently we have
\begin{align*}
\ln \begin{bmatrix}
\BX_{1,l+1}-\BX_1\\
\vdots\\
\BX_{s,l+1}-\BX_s
\end{bmatrix}\rn_F&\leq 2\ln
\begin{bmatrix}
\P_{\BQ_{1,l}}(\BX_{1,l})-\P_{\BQ_{1,l}}(\BX_{1})+\alpha_{1,l}\P_{\BQ_{1,l}}(\BG_{1,l})\\
\vdots\\
\P_{\BQ_{s,l}}(\BX_{s,l})-\P_{\BQ_{s,l}}(\BX_{s})+\alpha_{s,l}\P_{\BQ_{s,l}}(\BG_{s,l})
\end{bmatrix}
\rn_F\\
&=2\ln
\begin{bmatrix}
\P_{\BQ_{1,l}}(\BX_{1,l}-\BX_{1})-\alpha_{1,l}\P_{\BQ_{1,l}}\A_1^*\sum_{k=1}^s\A_k(\BX_{k,l}-\BX_k)\\
\vdots\\
\P_{\BQ_{s,l}}(\BX_{s,l}-\BX_{s})-\alpha_{s,l}\P_{\BQ_{1,l}}\A_s^*\sum_{k=1}^s\A_k(\BX_{k,l}-\BX_k)
\end{bmatrix}
\rn_F\\
&=2\ln
\begin{bmatrix}
\P_{\BQ_{1,l}}(\BX_{1,l}-\BX_{1})-\alpha_{1,l}\P_{\BQ_{1,l}}\A_1^*\sum_{k=1}^s\A_k\P_{\BQ_{k,l}}(\BX_{k,l}-\BX_k)\\
\vdots\\
\P_{\BQ_{s,l}}(\BX_{s,l}-\BX_{s})-\alpha_{s,l}\P_{\BQ_{1,l}}\A_s^*\sum_{k=1}^s\A_k\P_{\BQ_{k,l}}(\BX_{k,l}-\BX_k)
\end{bmatrix}
\rn_F\\
&:=2\ln\II_1\rn_F\numberthis\label{eq:iht_eq1}
\end{align*}
where in the third line, we have used the fact $\P_{\BQ_{k,l}}(\BX_{k,l}-\BX_k)=(\BX_{k,l}-\BX_k)$ for all $1\leq k\leq s$. In order to bound $\II_1$, we first rewrite it into the following matrix-vector product form
\begin{align*}
\II_1 =\A_\alpha
\begin{bmatrix}
\BX_{1,l}-\BX_{1}\\
\vdots\\
\BX_{s,l}-\BX_{s}
\end{bmatrix},\numberthis\label{eq:iht_eq2}
\end{align*}
where 
\begin{align*}\A_\alpha&=
\begin{bmatrix}
\P_{\BQ_{1,l}} & \cdots & 0\\
\vdots&\ddots&\vdots\\
0&\cdots&P_{\BQ_{s,l}}
\end{bmatrix}-\begin{bmatrix}
\alpha_{1,l}  & &\\
& \ddots&\\
& &  \alpha_{s,l}
\end{bmatrix}
\begin{bmatrix}
\P_{\BQ_{1,l}}\A_1^*\\
\vdots\\
\P_{\BQ_{s,l}}\A_s^*
\end{bmatrix}
\begin{bmatrix}
\A_1\P_{\BQ_{1,l}}&\cdots&\A_s\P_{\BQ_{s,l}}
\end{bmatrix}.
\end{align*}So it suffices to bound the spectral norm of $\A_\alpha$. Denote $\A_\alpha$ by $\A$ when $\alpha_{k,l}=1$ for all $1\leq k\leq s$ in $\A_{\alpha}$. Let $\BY=[\BY_1,\cdots,\BY_s]^T$. Since $\A$ is self-adjoint, we have 
\begin{align*}
\ln\A\rn&=\sup_{\ln\BY\rn_F=1}\lab\la\BY,\A(\BY)\ra\rab\\
& = \sup_{\ln\BY\rn_F=1}\lab\sum_{k=1}^s\ln\P_{\BQ_{k,l}}(\BY_k)\rn_F^2-\ln\sum_{k=1}^s\A_k\P_{\BQ_{k,l}}(\BY_k)\rn\rab\\
&\leq \sup_{\ln\BY\rn_F=1}\delta_{3r} \sum_{k=1}^s\ln\P_{\BQ_{k,l}}(\BY_k)\rn_F^2\\
&\leq \delta_{3r},
\end{align*}
where the third line follows from the ARIP and the fact  $\P_{\BQ_{k,l}}(\BY_k),~k=1,\cdots,s$ are matrices of rank at most $3r$. 
It follows that
\begin{align*}
\ln\A_\alpha\rn&\leq \ln\A_\alpha-\A\rn+\ln\A\rn\\
&\leq \max_{k}|\alpha_{k,l}-1|\ln\begin{bmatrix}
\P_{\BQ_{1,l}}\A_1^*\\
\P_{\BQ_{2,l}}\A_2^*\\
\vdots\\
\P_{\BQ_{s,l}}\A_s^*
\end{bmatrix}
\begin{bmatrix}
\A_1\P_{\BQ_{1,l}}&\A_2\P_{\BQ_{2,l}}&\cdots&\A_s\P_{\BQ_{s,l}}
\end{bmatrix}\rn\\
&\leq \frac{\delta_{2r}(1+\delta_{3r})}{1-\delta_{2r}}+\delta_{3r}\\
&\leq \frac{2\delta_{3r}}{1-\delta_{3r}},
\end{align*}
where in the third line we use the following ARIP bound for $\alpha_{k,l}$
\begin{align*}
\frac{1}{1+\delta_{2r}}\leq\alpha_{k,l}= \frac{\|\P_{T_{k,l}}( \BG_{k,l})\|_F^2}{\| \A_k\P_{T_{k,l}}( \BG_{k,l})\|_2^2}\leq \frac{1}{1-\delta_{2r}}.\numberthis\label{eq:RIP_alpha}
\end{align*}
Combining the spectral of $\A_\alpha$ together with \eqref{eq:iht_eq1} and \eqref{eq:iht_eq2} gives 
\begin{align*}
\ln \begin{bmatrix}
\BX_{1,l+1}-\BX_1\\
\vdots\\
\BX_{s,l+1}-\BX_s
\end{bmatrix}\rn_F&\leq \frac{4\delta_{3r}}{1-\delta_{3r}}
\ln \begin{bmatrix}
\BX_{1,l}-\BX_1\\
\vdots\\
\BX_{s,l}-\BX_s
\end{bmatrix}\rn_F,
\end{align*}
which completes the proof of the theorem.
\end{proof}
\subsection{Proof of Theorem~\ref{thm:iht_noise}}\label{ss:iht_noise}
\begin{proof}
To prove the convergence of IHT under additive noise, we only need to modify \eqref{eq:iht_eq1} slightly as follows:
\begin{align*}
\ln \begin{bmatrix}
\BX_{1,l+1}-\BX_1\\
\vdots\\
\BX_{s,l+1}-\BX_s
\end{bmatrix}\rn_F&\leq 2\ln
\begin{bmatrix}
\P_{\BQ_{1,l}}(\BX_{1,l})-\P_{\BQ_{1,l}}(\BX_{1})+\alpha_{1,l}\P_{\BQ_{1,l}}(\BG_{1,l})\\
\vdots\\
\P_{\BQ_{s,l}}(\BX_{s,l})-\P_{\BQ_{s,l}}(\BX_{s})+\alpha_{s,l}\P_{\BQ_{s,l}}(\BG_{s,l})
\end{bmatrix}
\rn_F\\
&\leq2\ln
\begin{bmatrix}
\P_{\BQ_{1,l}}(\BX_{1,l}-\BX_{1})-\alpha_{1,l}\P_{\BQ_{1,l}}\A_1^*\sum_{k=1}^s\A_k(\BX_{k,l}-\BX_k)\\
\vdots\\
\P_{\BQ_{s,l}}(\BX_{s,l}-\BX_{s})-\alpha_{s,l}\P_{\BQ_{1,l}}\A_s^*\sum_{k=1}^s\A_k(\BX_{k,l}-\BX_k)
\end{bmatrix}
\rn_F\\
&+2\ln
\begin{bmatrix}
\alpha_{1,l}\P_{\BQ_{1,l}}\A_{1}^*(\be)\\
\vdots\\
\alpha_{s,l}\P_{\BQ_{s,l}}\A_s^*(\be)
\end{bmatrix}
\rn_F\\
&:=\II'_{11}+\II'_{12}.
\end{align*}
The proof of Theorem~\ref{thm:iht} shows that 
\begin{align*}
\II'_{11}\leq \frac{4\delta_{3r}}{1-\delta_{3r}}
\ln \begin{bmatrix}
\BX_{1,l}-\BX_1\\
\vdots\\
\BX_{s,l}-\BX_s
\end{bmatrix}\rn_F,
\end{align*}
while $\II'_{12}$ can be bounded as 
\begin{align*}
\II'_{12}&\leq 2\max_{k}|\alpha_{k,l}|\ln\begin{bmatrix}
\P_{\BQ_{1,l}}\A_{1}^*(\be)\\
\vdots\\
\P_{\BQ_{s,l}}\A_s^*(\be)
\end{bmatrix}
\rn_F\\
&=2\max_{k}|\alpha_{k,l}|\max_{\{\BY_k\}_{k=1}^s,\sum_{k=1}^s\ln\BY_k\rn_F^2=1}\sum_{k=1}^s\la \P_{\BQ_{k,l}}\A_{1}^*(\be),\BY_k\ra\\
&\leq \frac{2}{1-\delta_{2r}} \ln\be\rn\max_{\{\BY_k\}_{k=1}^s,\sum_{k=1}^s\ln\BY_k\rn_F^2=1}\ln\sum_{k=1}^s\A_k\P_{Q_{k,l}}(\BY_k)\rn\\
&\leq \frac{2\sqrt{1+\delta_{3r}}}{1-\delta_{2r}}\ln\be\rn,
\end{align*}
where in the third line we utilize the bound for $\alpha_{k,l}$ in \eqref{eq:RIP_alpha} and in the last line we utilize the definition of the ARIP.  Therefore,
\begin{align*}
\ln \begin{bmatrix}
\BX_{1,l+1}-\BX_1\\
\vdots\\
\BX_{s,l+1}-\BX_s
\end{bmatrix}\rn_F\leq \frac{4\delta_{3r}}{1-\delta_{3r}}
\ln \begin{bmatrix}
\BX_{1,l}-\BX_1\\
\vdots\\
\BX_{s,l}-\BX_s
\end{bmatrix}\rn_F+\frac{2\sqrt{1+\delta_{3r}}}{1-\delta_{2r}}\ln\be\rn. 
\end{align*}
The proof is complete after we apply the above inequality recursively.
\end{proof}
\subsection{Proof of Theorem~\ref{thm:fiht}}\label{ss:fiht}
\begin{proof}
Let $\BW_{k,l} = \BX_{k,l}+\alpha_{k,l}\P_{T_{k,l}}(\BG_{k,l})$. The following inequality holds:
\begin{align*}
\ln \begin{bmatrix}
\BX_{1,l+1}-\BX_1\\
\vdots\\
\BX_{s,l+1}-\BX_s
\end{bmatrix}\rn_F
=\sqrt{\sum_{k=1}^s\ln\BX_{k,l+1}-\BX_k\rn_F^2}\leq
\sqrt{4\sum_{k=1}^s\ln \BW_{k,l}-\BX_k\rn_F^2}=
2\ln \begin{bmatrix}
\BW_{1,l}-\BX_1\\
\vdots\\
\BW_{s,l}-\BX_s
\end{bmatrix}\rn_F,
\end{align*}
where the inequality follows from 
\begin{align*}
\ln\BX_{k,l+1}-\BX_k\rn_F\leq \ln\BX_{k,l+1}-\BW_{k,l}\rn_F+\ln\BW_{k,l}-\BX_k\rn_F\leq 2\ln \BW_{k,l}-\BX_k\rn_F.
\end{align*}
This leads to 
\begin{align*}
\ln \begin{bmatrix}
\BX_{1,l+1}-\BX_1\\
\vdots\\
\BX_{s,l+1}-\BX_s
\end{bmatrix}\rn_F&\leq 2\ln
\begin{bmatrix}
\BX_{1,l}+\alpha_{1,l}\P_{T_{1,l}}(\BG_{1,l})-\BX_1\\
\vdots\\
\BX_{s,l}+\alpha_{s,l}\P_{T_{1,l}}(\BG_{s,l})-\BX_s
\end{bmatrix}
\rn_F\\&=
2\ln\begin{bmatrix}
\BX_{1,l}-\BX_1-\alpha_{1,l}\P_{T_{1,l}}\A_1^*\sum_{k=1}^s\A_k(\BX_{k,l}-\BX_k)\\
\vdots\\
\BX_{s,l}-\BX_s-\alpha_{s,l}\P_{T_{s,l}}\A_s^*\sum_{k=1}^s\A_k(\BX_{k,l}-\BX_k)
\end{bmatrix}
\rn_F\\
&\leq 2\ln
\begin{bmatrix}
\P_{T_{1,l}}(\BX_{1,l}-\BX_1)-\alpha_{1,l}\P_{T_{1,l}}\A_1^*\sum_{k=1}^s\A_k\P_{T_{k,l}}(\BX_{k,l}-\BX_k)\\
\vdots\\
\P_{T_{s,l}}(\BX_{s,l}-\BX_s)-\alpha_{s,l}\P_{T_{s,l}}\A_s^*\sum_{k=1}^s\A_k\P_{T_{k,l}}(\BX_{k,l}-\BX_k)
\end{bmatrix}
\rn_F\\
&+2\ln
\begin{bmatrix}
\alpha_{1,l}\P_{T_{1,l}}\A_1^*\sum_{k=1}^s\A_k(\I-\P_{T_{k,l}})(\BX_{k,l}-\BX_k)\\
\vdots\\
\alpha_{s,l}\P_{T_{s,l}}\A_s^*\sum_{k=1}^s\A_k(\I-\P_{T_{k,l}})(\BX_{k,l}-\BX_k)
\end{bmatrix}
\rn_F\\
&+2\ln
{\begin{bmatrix}
(\I-\P_{T_{1,l}})\BX_1\\
\vdots\\
(\I-\P_{T_{s,l}})\BX_s
\end{bmatrix}}
\rn_F\\
&:=2\ln\II_2\rn_F+2\ln\II_3\rn_F+2\ln\II_4\rn_F.\numberthis\label{eq:fiht_main}
\end{align*}

Following the same argument for the bound of $\II_1$ in \eqref{eq:iht_eq1}, we can bound $\II_2$ as
\begin{align*}
\ln\II_2\rn_F\leq \frac{2\delta_{2r}}{1-\delta_{2r}}\ln \begin{bmatrix}
\BX_{1,l}-\BX_1\\
\vdots\\
\BX_{s,l}-\BX_s
\end{bmatrix}\rn_F
\end{align*}
by noting that all the matrices in $T_{k,l}$ are of rank at most $2r$.

Next, $\II_3$ can be bounded in the following way:
\begin{align*}
\ln\II_3\rn_F&\leq\max_{k}|\alpha_{k,l}|\sup_{\{\BY_k\}_{k=1}^s,\sum_{k=1}^s\ln\BY_k\rn_F^2=1}\lab\sum_{k=1}^s\la\BY_k,\P_{T_{k,l}}\A_k^*\sum_{k=1}^s\A_k(\I-\P_{T_{k,l}})(\BX_{k,l}-\BX_k)\ra\rab\\
&=\max_{k}|\alpha_{k,l}|\sup_{\{\BY_k\}_{k=1}^s,\sum_{k=1}^s\ln\BY_k\rn_F^2=1}\lab\sum_{k=1}^s\la\A_k\P_{T_{k,l}}\lb\BY_k\rb,\sum_{k=1}^s\A_k(\I-\P_{T_{k,l}})(\BX_{k,l}-\BX_k)\ra\rab\\
&\leq \max_{k}|\alpha_{k,l}|\sup_{\{\BY_k\}_{k=1}^s,\sum_{k=1}^s\ln\BY_k\rn_F^2=1}\delta_{3r}\sqrt{\sum_{k=1}^s\ln\P_{T_{k,l}}(\BY_k)\rn_F^2}
\sqrt{\sum_{k=1}^s\ln (\I-\P_{T_{k,l}})(\BX_{k,l}-\BX_k)\rn_F^2}\\
&\leq \frac{\delta_{3r}}{1-\delta_{2r}}\ln \begin{bmatrix}
\BX_{1,l}-\BX_1\\
\vdots\\
\BX_{s,l}-\BX_s
\end{bmatrix}\rn_F,
\end{align*}
where the third line follows from Lemma~\ref{lem:aux1} in the appendix, and the fourth line follows from the ARIP bound for $\alpha_{k,l}$ in \eqref{eq:RIP_alpha}. 

To bound $\II_4$, first note that the application of Lemma~\ref{lem:aux2} in the appendix gives
\begin{align*}
\|(\I-\P_{T_{k,l}})(\BX_k)\|_F\leq \frac{\|\BX_{k,l}-\BX_k\|_F^2}{\sigmin(\BX_k)}\leq \frac{\max_k\ln\BX_{k,l}-\BX_k\rn_F}{\sigmin}\ln\BX_{k,l}-\BX_k\rn_F,
\end{align*}
where $\sigmin:=\min_{k}\sigmin(\BX_k)$.
Thus,
\begin{align*}
\ln\II_4\rn_F\leq \frac{\max_k\ln\BX_{k,l}-\BX_k\rn_F}{\sigmin}\ln \begin{bmatrix}
\BX_{1,l}-\BX_1\\
\vdots\\
\BX_{s,l}-\BX_s
\end{bmatrix}\rn_F\leq 
\frac{1}{\sigmin}\ln \begin{bmatrix}
\BX_{1,l}-\BX_1\\
\vdots\\
\BX_{s,l}-\BX_s
\end{bmatrix}\rn_F\cdot\ln \begin{bmatrix}
\BX_{1,l}-\BX_1\\
\vdots\\
\BX_{s,l}-\BX_s
\end{bmatrix}\rn_F.
\end{align*}

Substituting the bounds for $\II_2$, $\II_3$ and $\II_4$ into \eqref{eq:fiht_main} gives
\begin{align*}
\ln \begin{bmatrix}
\BX_{1,l+1}-\BX_1\\
\vdots\\
\BX_{s,l+1}-\BX_s
\end{bmatrix}\rn_F&\leq 2\lb\frac{2\delta_{2r}}{1-\delta_{2r}}+\frac{\delta_{3r}}{1-\delta_{2r}}+\frac{1}{\sigmin}\ln \begin{bmatrix}
\BX_{1,l}-\BX_1\\
\vdots\\
\BX_{s,l}-\BX_s
\end{bmatrix}\rn_F\rb\ln \begin{bmatrix}
\BX_{1,l}-\BX_1\\
\vdots\\
\BX_{s,l}-\BX_s
\end{bmatrix}\rn_F.
\end{align*}
Assume
\begin{align*}
\rho:=2\lb\frac{2\delta_{2r}}{1-\delta_{2r}}+\frac{\delta_{3r}}{1-\delta_{2r}}+\frac{1}{\sigmin}\ln \begin{bmatrix}
\BX_{1,0}-\BX_1\\
\vdots\\
\BX_{s,0}-\BX_s
\end{bmatrix}\rn_F\rb < 1.
\end{align*}
By  induction, we have 
\begin{align*}
\ln \begin{bmatrix}
\BX_{1,l+1}-\BX_1\\
\vdots\\
\BX_{s,l+1}-\BX_s
\end{bmatrix}\rn_F&\leq \rho\ln \begin{bmatrix}
\BX_{1,l}-\BX_1\\
\vdots\\
\BX_{s,l}-\BX_s
\end{bmatrix}\rn_F
\end{align*}
Note that the initial guess is obtained by one-step hard thresholding with $\alpha_{k,l}=1$. So  following the proof of Theorem~\ref{thm:iht} we have 
\begin{align*}
\ln \begin{bmatrix}
\BX_{1,0}-\BX_1\\
\vdots\\
\BX_{s,0}-\BX_s
\end{bmatrix}\rn_F\leq 2\delta_{3r}
\ln \begin{bmatrix}
\BX_{1}\\
\vdots\\
\BX_{s}
\end{bmatrix}\rn_F.
\end{align*}
Therefore, one has $\rho<1$ if 
\begin{align*}
2\lb\frac{2\delta_{2r}}{1-\delta_{2r}}+\frac{\delta_{3r}}{1-\delta_{2r}}+\frac{2\delta_{3r}}{\sigmin}\ln \begin{bmatrix}
\BX_1\\
\vdots\\
\BX_s
\end{bmatrix}\rn_F\rb<1.\numberthis\label{eq:fiht_eq1}
\end{align*}
Since 
\begin{align*}
\ln \begin{bmatrix}
\BX_1\\
\vdots\\
\BX_s
\end{bmatrix}\rn_F=\sqrt{\sum_{k=1}^s\ln\BX_k\rn_F^2}\leq\sqrt{\sum_{k=1}^sr\sigmax^2(\BX_k)}\leq \sqrt{rs}\sigmax,
\end{align*}
\eqref{eq:fiht_eq1} holds if 
\begin{align*}
2\lb\frac{2\delta_{2r}}{1-\delta_{2r}}+\frac{\delta_{3r}}{1-\delta_{2r}}+2\delta_{3r}\sqrt{rs}\frac{\sigmax}{\sigmin}\rb<1,
\end{align*}
where $\sigmax:=\max_k\sigmax(\BX_k)$.
\end{proof}

\section{Conclusion and Future Directions}
\label{s:conclusion}

We have presented the first computationally efficient algorithms that can extract low rank matrices from a sum of linear measurements.
These algorithms have potential applications in areas  such as wireless communications and quantum tomography.
Numerical simulations show an empirical performance that is quite close to the information theoretic limit in terms of
demixing ability with respect to the number of measurements. 

At the same time, there are still a number of open
questions for further directions. Firstly,  the robustness analysis of FIHT needs to be addressed in the future. Secondly,  our theoretical framework so far is still a bit restrictive, since it only yields
close-to-optimal results for Gaussian measurement matrices, which are however of limited use in applications.
Thus, one future challenge consist in establishing good Amalgam-RIP bounds for structured measurement matrices.
Thirdly, it is also interesting to see whether  the Amalgam-RIP  can be used to analyse other approaches for  low rank matrix demixing. In particular, we want to investigate whether the Amalgam form of RIP in \eqref{eq:ARIP_def}, but restricted onto a local subspace, is sufficient for the  guarantee analysis of the nuclear norm minimization studied in \cite{LS15}. Finally, in this paper the Amalgam-RIP is defined for homogeneous data, i.e.,  matrices of low rank. It is likely that similar Amalgam-RIP can be established for  heterogeneous data and then be used in the analysis of different reconstruction programs.

\section*{Appendix}\label{s:appendix}
\begin{lemma}\label{lem:aux1}
Suppose $\la\BY_k,\BZ_k\ra=0$ and $\rank(\BY_k+\BZ_k)\leq c$ for all $1\leq k\leq s$. Then 
\begin{align*}
\lab\la \sum_{k=1}^s\A_k(\BY_k),\sum_{k=1}^s\A_k(\BZ_k)\ra\rab\leq\delta_c\sqrt{\sum_{k=1}^s\ln\BY_k\rn_F^2}\sqrt{\sum_{k=1}^s\ln\BZ_k\rn_F^2}.
\end{align*}
\end{lemma}
\begin{proof}
Due to the homogeneity, we can assume $\sqrt{\sum_{k=1}^s\ln\BY_k\rn_F^2}=1,~\sqrt{\sum_{k=1}^s\ln\BZ_k\rn_F^2}=1$.
Since 
$\la\BY_k,\BZ_k\ra=0$, we have $\sum_{k=1}^s\ln\BY_k+\BZ_k\rn_F^2 = 2$. So
\begin{align*}
2(1-\delta_c)\leq \ln\sum_{k=1}^s\A_k(\BY_k\pm\BZ_k)\rn^2\leq 2(1+\delta_c)
\end{align*}
following from Def.~\ref{def:rip}. Then the application of the parallelogram identity implies 
\begin{align*}
\lab\la \sum_{k=1}^s\A_k(\BY_k),\sum_{k=1}^s\A_k(\BZ_k)\ra\rab=\frac{1}{4}\lab \ln\sum_{k=1}^s\A_k(\BY_k+\BZ_k)\rn^2-
\ln\sum_{k=1}^s\A_k(\BY_k-\BZ_k)\rn^2\rab\leq\delta_c,
\end{align*}
which concludes the proof.
\end{proof}
\begin{lemma}[{\cite[Lem.~4.1]{CGIHT_dense}}]\label{lem:aux2}
Let $\BX_{k,l}=\BU_{k,l}\BS_{k,l}\BV^*_{k,l}$ be rank $r$ matrix and $T_{k,l}$ be the tangent space of the rank $r$ matrix manifold at $\BX_{k,l}$. Let $\BX_k$ be another rank $r$ matrix. One has 
\begin{align*}
\|(\I-\P_{T_{k,l}})(\BX_k)\|_F\leq \frac{\|\BX_{k,l}-\BX_k\|_F^2}{\sigmin(\BX_k)}.
\end{align*}
\end{lemma}

\bibliography{demix,optim}
\bibliographystyle{siam}



\end{document}